\documentclass[a4paper,11pt]{article}
 \usepackage[english]{babel}
\usepackage[a4paper]{geometry}
\usepackage{amsmath}
\usepackage{amsthm}
\usepackage{amssymb}
\usepackage{mathtools}
\usepackage{bbm}
\usepackage{color}
\usepackage{comment}
\usepackage{mathrsfs}
\usepackage{enumerate}
\usepackage{bm}

\usepackage{hyperref}         

\numberwithin{equation}{section}

\def\eps{\varepsilon}

\DeclareMathOperator{\sech}{sech}

\newcommand{\ee}{{\mathbb{E}\,}}

\newtheorem{theorem}{Theorem}[section]  
\newtheorem{cor}[theorem]{Corollary}
\newtheorem{prop}[theorem]{Proposition}
\newtheorem{lemma}[theorem]{Lemma}
\newtheorem{remark}[theorem]{Remark}

        \let\k=\kappa     
                          \let\r=r

\begin{document}

\title{Dynamical Approach to the TAP Equations for the Sherrington-Kirkpatrick Model}

\author{Arka Adhikari\footnote{adhikari@math.harvard.edu}, Christian Brennecke\footnote{brennecke@math.harvard.edu}, Per von Soosten\footnote{vonsoosten@math.harvard.edu}, Horng-Tzer Yau\footnote{htyau@math.harvard.edu}
\\
\\
Department of Mathematics, Harvard University, \\
One Oxford Street, Cambridge MA 02138, USA$$ \\
\\}

\date{February 19, 2021}  

\maketitle

\begin{abstract}
We present a new dynamical proof of the Thouless--Anderson--Palmer (TAP) equations for the classical Sherrington-Kirkpatrick spin glass at sufficiently high temperature. In our derivation, the TAP equations are a simple consequence of the decay of the two point correlation functions. The methods can also be used to establish the decay of higher order correlation functions. We illustrate this by proving a suitable decay bound on the three point functions from which we derive an analogue of the TAP equations for the two point functions.
\end{abstract}

\section{Introduction}

We consider systems of $N$ spins $ \sigma_i$, $i\in\{1,\dots, N\}$, taking values in $\{-1,1\}$. The Hamiltonian $H_N:\{-1,1\}^N\to\mathbb{R}$ of the system is defined by
		\begin{equation} \label{eq:HN}
		H_N(\sigma)= H_N(\sigma_1,\dots,\sigma_N) = \sum_{1\leq i< j\leq N} g_{ij} \sigma_{i}\sigma_j + h\sum_{i=1}^N  \sigma_i,
		\end{equation}
where the couplings $ \{g_{ij}\}$ are i.i.d. Gaussians of variance $t/N$ and $h\in\mathbb{R}$ denotes the external field strength. For definiteness, we also set $g_{ii}=0$ for all $i\in\{1,\dots, N\}$. In our setup $t = \beta^2$ plays the role of the inverse temperature, but the present notation will be more natural in the dynamical context we consider in the sequel.

The Hamiltonian \eqref{eq:HN} corresponds to the classical Sherrington--Kirkpatrick (SK) spin glass model \cite{SK}. The understanding of basic thermodynamic quantities of this model has required significant efforts by many physicists and mathematicians. In particular, the famous Parisi formula \cite{Par1, Par2} for the free energy in the thermodynamic limit was proved by Guerra \cite{Gue} and Talagrand \cite{Tal}. Later, the ultrametricity \cite{Pan1} was established by Panchenko \cite{Pan3} for generic models. We refer to the standard works \cite{MPV, Tal1, Tal2, Pan2} for a thorough introduction to the SK and more general spin glass models and for a comprehensive list of references.

In this paper, we are concerned with the magnetizations and two-point correlation functions defined by
		\[m_i = \langle  \sigma_i \rangle, \hspace{0.5cm} m_{ij} = \langle \sigma_i \sigma_j \rangle - \langle \sigma_i \rangle \langle \sigma_j \rangle,\]
where
		\[\langle f \rangle  = \frac1{Z_N} \sum_{\sigma\in\{-1,1\}^N} f(\sigma)\, e^{H_N(\sigma)}, \hspace{0.5cm} Z_N =  \sum_{\sigma\in\{-1,1\}^N}  e^{H_N(\sigma)}\]
denotes the Gibbs expectation. At high temperature, the Thouless--Anderson--Palmer (TAP) equations \cite{TAP} predict that the magnetizations satisfy the system of self-consistent equations
\begin{equation}\label{eq:TAPsimeq} m_i \approx \tanh\Big( h + \sum_{k\neq i} g_{ik}m_k - t(1-q)m_i\Big)\end{equation}
in a sense that will be made precise later. In~\eqref{eq:TAPsimeq},  $q =q(t,h)$ is the solution of the fixed-point equation $q = \mathbb{E} \tanh^2(\sqrt{tq} Z +h)$ where $Z \sim \mathcal{N}(0,1)$ is a standard Gaussian random variable. Physically, the value $q\in [0;1]$ corresponds to the limiting value of the overlap distribution in the replica-symmetric high temperature regime. The overlap $ R_{1,2}:\{-1,1\}^N\times \{-1,1\}^N\to \mathbb{R}$ is defined by
		\[R_{1,2}(\sigma^1, \sigma^2) = \frac1N\sum_{i=1}^N \sigma_i^1\sigma_i^2.\]
Its distribution under $ \mathbb{E} \langle\cdot\rangle\otimes \langle\cdot\rangle$ is the functional order parameter of the system in the thermodynamic limit $N\to\infty$ . At sufficiently high temperature, the overlap distribution is expected to concentrate on a single point. 
In fact, for $t<1/4$, one can prove that the overlap concentrates exponentially, which is a key input for a detailed mathematical understanding of the Gibbs measure at high temperature (see \cite[Sections 1.4 to 1.11]{Tal1}).

The validity of the TAP equations \eqref{eq:TAPsimeq} at high temperature has been established by Talagrand \cite{Tal0, Tal1} and Chatterjee \cite{Cha}. Both works rely on the concentration of the overlap as a key ingredient in the proof. More recently, Bolthausen \cite{Bolt1, Bolt2} constructed an iterative solution of the TAP equations in the full high temperature regime and used it to provide a new proof of the replica-symmetric formula for the free energy at sufficiently high temperature. In \cite{ChTa}, Chen and Tang proved that Bolthausen's scheme indeed approximates the magnetizations of the SK model, assuming locally uniform concentration of the overlap. At low temperature, Auffinger and Jagannath \cite{AuJag} proved a version of the TAP equations for generic mixed $p$-spin models. In this case, the overlap is not a constant anymore, but one can decompose the hypercube into clusters (``pure states") within which the overlap remains approximately constant. Then, the TAP equations remain valid conditionally on each cluster (see \cite{AuJag} for more precise details).

An interesting open problem is to prove the replica-symmetry of the SK model in the full high temperature regime predicted by de Almeida and Thouless \cite{AT}. The system is believed to be replica-symmetric for all $(t,h)$ that satisfy
		\begin{equation}\label{eq:ATline}  \mathbb{E} \frac{t}{\cosh^4(\sqrt{tq}Z+h)} <1, \end{equation}
where $q = \mathbb{E} \tanh^2(\sqrt{tq} Z +h)$ and $Z \sim\mathcal{N}(0,1)$ as above. In particular, the TAP equations \eqref{eq:TAPsimeq} are believed to be valid under the AT condition \eqref{eq:ATline}. So far, replica-symmetry is known above the AT line up to a bounded region in the $(t,h)$-phase diagram. This has been proved in \cite{JagTob} through an analysis of the Parisi variational problem.

The goal of this work is to present a new proof of the TAP equations that relies on a direct dynamical approach by viewing the couplings $g_{ij}$ as Brownian motions running at speed $1/N$. After applying It\^o's lemma to the magnetizations, this point of view leads naturally to a dynamical study of the two point functions $ m_{ij}$. For sufficiently high temperature, we prove suitable decay bounds on the $m_{ij}$ from which the TAP equations follow with explicit error bounds as a simple corollary. Our approach extends to higher order correlation functions in a straightforward way. In particular, we prove an analogue of the TAP equations for the two point functions which provides a simple heuristic connection to the AT condition \eqref{eq:ATline}. For this reason, we hope that a dynamical approach will contribute to an improved understanding of the high temperature regime.

Tools from stochastic calculus have provided useful insights into the probabilistic structure of the SK model in the past. Comets and Neveu \cite{CoNe} gave an elegant new proof of the fundamental high temperature results of Aizenman, Lebowitz and Ruelle \cite{AizLeRue} in the absence of an external field by representing the partition function as a suitable stochastic exponential and invoking a martingale central limit theorem. Moreover, the interpolation method of Guerra, whose core mechanism is based on Gaussian integration by parts, can also be rewritten dynamically in terms of It\^{o}'s lemma. The paper of Tindel \cite{Tin} combines the previous two perspectives to extend the central limit theorem for the free energy to a region with positive external field strength. In contrast to these works, our present approach directly tracks the evolution of the magnetization and higher order correlation functions as the coupling strengths between one particle and the others are gradually increased. This approach gives rise to the TAP equations in a natural fashion and makes the corresponding computations for the higher order correlation functions systematic and tractable.

For the statement of our main results, let $ m_k^{(i)}$ and $m_{kl}^{(i)}$ denote the magnetizations and two point correlation functions, respectively, after the $i$-th particle $\sigma_i$ has been removed from the $N$-spin system (see the next section for a precise definition). Our main result describes the validity of a hierarchical version of the TAP equations (also called the cavity equations) for all $0\leq t<\log 2$ in the sense of $L^2(\mathbb{P})$.

\begin{theorem}\label{thm:main}
Let $0\leq t<\log 2$. Then, there exists a constant $C=C_t>0$, independent of $N\in\mathbb{N}$, such that	
		\begin{equation}\label{eq:hTAP1} \mathbb{E} \bigg[m_i - \tanh\Big( h + \sum_{j\neq i} g_{ij}m_j^{(i)}\Big)\bigg]^2\leq \frac{C}N. \end{equation}
Moreover, for all $\epsilon>0$ sufficiently small and $i\neq j$, there exists $C=C_{t,\epsilon}>0$ such that
		\begin{equation}\label{eq:hTAP2}  \mathbb{E} \bigg[ m_{ij} -  \bigg(1- \tanh^2\Big(h + \sum_{k\neq i} g_{ik}m_k^{(i)}\Big)\bigg) \sum_{l\neq i} g_{il}m_{lj}^{(i)}  \bigg]^2\leq \frac{C}{N^{1+\epsilon}}.  \end{equation}
\end{theorem}

We point out that equation~\eqref{eq:hTAP1} for the magnetizations has been studied before in~\cite[Lemma 1.7.4]{Tal1}, where a similar bound is proved for $t<1/4$. In fact,~\eqref{eq:hTAP1} is what one would expect from the classical heuristic
\[m_i = \frac{\Big\langle \sinh\left(h + \sum_{j \neq i} g_{ij} \sigma_j \right) \Big\rangle^{(i)}}{\Big\langle \cosh \left(h + \sum_{j \neq i} g_{ij} \sigma_j \right)\Big\rangle^{(i)}} \approx \tanh \bigg(h + \sum_{j \neq i} g_{ij} m_j^{(i)}  \bigg)\]
for a mean-field ferromagnet, which is correct (at least) when the spins are approximately independent under the Gibbs measure. However, unlike the ferromagnetic case, the typical size of the couplings $g_{ij} = \mathcal{O}(N^{-1/2})$ and the correlations between $g_{ij}$ and $m^{(i)}_j$ prohibit one from obtaining the classical mean--field equations by inserting the heuristic $m^{(i)}_j \approx m_j$. Instead, this substitution results in the Onsager correction $ t(1-q)m_i$ in the TAP equations. The significance of \eqref{eq:hTAP1} and \eqref{eq:hTAP2} is that they display the leading order dependence of $m_i$ and $m_{ij}$ on the $i$-th column $(g_{ik})_{1\leq k\leq N}$ of the interaction. Notice that, on a heuristic level, the equations \eqref{eq:hTAP2} for the $m_{ij}$ follow simply by differentiation of the TAP equations \eqref{eq:hTAP1} for the $m_i$ with respect to the external field. Alternatively, \eqref{eq:hTAP2} can also be derived using a cavity field heuristic, see \cite[Section V.3]{MPV}.

As already observed in \cite[Section V.3]{MPV}, it is interesting to note that the hierarchical TAP equations for the one and two point functions have a simple connection to the AT condition \eqref{eq:ATline}. To see this, let us assume that
		\[ q_N = \frac1N\sum_{k=1}^N m_k^2 \approx \frac1N\sum_{k=1}^N \big(m_k^{(i)}\big)^2 = q_N^{(i)},\]
which follows from the decay of correlations and let us assume in addition that 
		\begin{equation}\label{eq:concass} q_N = \frac1N\sum_{k=1}^N m_k^2 \approx \mathbb{E} \frac1N\sum_{k=1}^N m_k^2. \end{equation}
Notice that this concentration assumption is reasonable since
		\[ q_N = \frac1N \sum_{k=1}^N m_k^2 = \langle R_{1,2}\rangle.\]
We then conclude from the TAP equations \eqref{eq:hTAP1} and \eqref{eq:concass} that
		\[
		q_N \approx \mathbb{E} \tanh^2 \Big( h + \sum_{j\neq i} g_{ij}m_j^{(i)}\Big) = \mathbb{E}\tanh^2\Big(h +\sqrt{t q_N^{(i)}} Z_i\Big)\approx\mathbb{E}\tanh^2(h +\sqrt{t q_N} Z_i)
		\]
for the standard Gaussian $Z_i=\big(t q_N^{(i)}\big)^{-1/2} \sum_{k\neq i} g_{ik}m_k^{(i)} \sim \mathcal{N}(0,1)$. Hence, we expect that $q_N\approx q $ is close to the unique fixed point $ q = \mathbb{E}\tanh^2(h +\sqrt {tq}Z)$. Based on Theorem \ref{thm:main}, we will make this rigorous and prove the following concentration result.
\begin{prop}\label{prop:concentration}
Let $ 0\leq t<\log 2$ and let $ q = \mathbb{E}\tanh^2(h +\sqrt {tq}Z)$, where $Z\sim \mathcal{N}(0,1)$ denotes a standard Gaussian random variable. Let $q_N = N^{-1}\sum_{k=1}^N m_k^2$, then there exists a constant $C=C_t>0$ such that
		\begin{equation}\label{eq:concentration}\mathbb{E}\, | q - q_N|^2 \leq \frac{C}{N^{1/2}}.\end{equation}
\end{prop}
If we use the information of Proposition \ref{prop:concentration} and assume in addition that mixed moments of distinct correlation functions are of lower order $o(N^{-1})$, we recover the AT transition line \eqref{eq:ATline} as a singularity in the norm of the two point functions. More precisely, applying \eqref{eq:hTAP2}, we obtain from Gaussian integration by parts and separating the diagonal term in the sum $ \sum_{l\neq i} \big(m_{lj}^{(i)}\big)^2 $ that
		\begin{equation}\label{eq:heuristic}\begin{split}
		\mathbb{E}\, m_{ij}^2 \approx &\, \mathbb{E} \,t \bigg(1- \tanh^2\Big(h + \sum_{k\neq i} g_{ik}m_k^{(i)}\Big)\bigg)^2 \frac1N \sum_{l\neq i} \big(m_{lj}^{(i)}\big)^2 \\
		&\, + \mathbb{E}\frac{t^2}{N^2}\sum_{l_1, l_2\neq i}  \bigg[\partial_{il_1}\partial_{il_2} \bigg(1- \tanh^2\Big(h + \sum_{k\neq i} g_{ik}m_k^{(i)}\Big)\bigg)^2\bigg] m_{l_1j}^{(i)}m_{l_2j}^{(i)} \\
		\approx &\;  \mathbb{E} \,t\big[1- \tanh^2(h +  \sqrt{tq} Z\big)\big]^2\, \bigg[\frac1N \mathbb{E} \Big(1- \big(m_j^{(i)}\big)^2\Big)^2+ \mathbb{E}\frac1N\sum_{l\neq i,j} \big(m_{lj}^{(i)}\big)^2\bigg] + o(N^{-1})\\
		\approx &\; \frac{t}N\bigg[\mathbb{E} \frac{1}{\cosh^4(h +  \sqrt{tq} Z\big)}\bigg]^2 + \mathbb{E} \frac{t}{\cosh^4(h +  \sqrt{tq} Z\big)}\,\mathbb{E} \, m_{ij} ^2.
		\end{split}\end{equation}
Here, we used the approximation $ \mathbb{E}\,\big(m_{lj}^{(i)}\big)^2 \approx \mathbb{E} \, m_{lj} ^2$, which will be justified later. Moreover, we used that
\[Z_i = \big(t q_N^{(i)}\big)^{-1/2} \sum_{k\neq i} g_{ik}m_k^{(i)}\]
is independent of the remaining disorder $ g_{kl} $, for $k,l \neq i$, because of the Gaussian structure (see also \cite[Lemma 1.7.6]{Tal1}). Altogether, we expect that
		\[ \lim_{N\to \infty }\mathbb{E}\, \big( \sqrt Nm_{ij}\big)^2  = t \bigg[1-\mathbb{E} \frac{t}{\cosh^4(h +  \sqrt{tq} Z\big)}\bigg]^{-1}\bigg[\mathbb{E} \frac{1}{\cosh^4(h +  \sqrt{tq} Z\big)}\bigg]^2,  \]
where the right hand side is finite if \eqref{eq:ATline} holds true. Based on \eqref{eq:heuristic} as well as the results of Theorem \ref{thm:main} and Proposition \ref{prop:concentration}, we will prove the following proposition.
\begin{prop}\label{prop:mij2norm} Let $ 0\leq t<\log 2$ and let $ q = \mathbb{E}\tanh^2(h +\sqrt {tq}Z)$, where $Z\sim \mathcal{N}(0,1)$ denotes a standard Gaussian random variable. Then, for every $\epsilon>0$ sufficiently small, there exists a constant $C= C_{\eps,t}>$ so that
		\begin{equation}\label{eq:2normmij} \mathbb{E} \, m_{ij}^2 = \frac{t}N \bigg[1-\mathbb{E} \frac{t}{\cosh^4(h +  \sqrt{tq} Z\big)}\bigg]^{-1}\bigg[\mathbb{E} \frac{1}{\cosh^4(h +  \sqrt{tq} Z\big)}\bigg]^2 + \Theta \end{equation}
for an error $\Theta$ bounded by $| \Theta| \leq C/N^{1+\epsilon}$.
\end{prop}
The leading order behavior \eqref{eq:2normmij} of the two point functions $m_{ij}$ is well-known and already mentioned in \cite[Section V.3]{MPV}. A rigorous proof of the identity \eqref{eq:2normmij} for $t<1/4$ can be found in \cite[Section 1.8]{Tal1} and higher moments of the $m_{ij}$ were analyzed in \cite{Han}. These proofs are, however, not based on the heuristics outlined in \eqref{eq:heuristic}. 

As a corollary of Theorem \ref{thm:main}, we are also able to derive the TAP equations.
\begin{cor}\label{cor:main}
Let $0\leq t<\log 2$. Then, there exists a constant $C=C_t>0$, independent of $N\in\mathbb{N}$, such that	
		\begin{equation}\label{eq:TAP1} \mathbb{E} \bigg[m_i - \tanh\Big( h + \sum_{j\neq i} g_{ij}m_j - t(1-q_N) m_i\Big)\bigg]^2\leq \frac{C}N, \end{equation}
where $q_N$ is defined by $q_N = N^{-1}\sum_{k=1}^N m_k^2$.

Moreover, for any $\epsilon>0$ sufficiently small, there exists $C=C_{t,\epsilon}>0$ such that
		\begin{equation}\label{eq:TAP2}
		\begin{split}
		  &\mathbb{E} \bigg[ m_{ij} -   \big(1-m_i^2\big) \bigg(\sum_{k\neq i} g_{ik}m_{kj}+\frac {2t}N (M m)_j m_i-t(1-q_N)m_{ij} \bigg)  \bigg]^2\leq \frac{C}{N^{1+\epsilon}}
		\end{split}\end{equation}
for all $i\neq j$. Here, we set $ M = (m_{kl})_{1\leq k,l\leq N}$ and $m = (m_1,\dots, m_N)$.
\end{cor}
We point out that, using Proposition \ref{prop:concentration}, we can replace $q_N$ in \eqref{eq:TAP1} by the solution $q=\mathbb{E}\tanh^2(h+\sqrt{tq}Z)$, up to another error that vanishes as $N\to\infty$. This yields \eqref{eq:TAP1} in the form that is typical in the mathematical literature on the subject.

\begin{remark}
Let us mention that \eqref{eq:TAP2} represents a resolvent equation for the matrix $M=(m_{kl})_{1\leq k,l\leq N}$. Indeed, neglecting the error terms, \eqref{eq:TAP2} means that
\begin{equation} \label{eq:resheuristic} M \approx \frac{1}{ \Lambda - tA -G  - E_0 },
\end{equation}
where $ \Lambda_{ij} = (1-m_i^2)^{-1}\delta_{ij}, \hspace{0.5cm}A_{ij}= 2N^{-1} m_im_j ,  \hspace{0.5cm}E_0 = -t(1-q_N)\approx -t(1-q)$ and $G$ consists of the couplings $\{g_{ij}\}$ extended to a symmetric matrix. Thus, one recovers the resolvent of a deformed Gaussian Orthogonal Ensemble at the energy $E_0$. Like the heuristics following Theorem \ref{thm:main}, this suggests to study the high temperature regime in view of the singularity of $M$, a viewpoint reminiscent of \cite{Ple} (see also \cite[Eq. (3.3)]{Ple}). 

Based on the observation in \eqref{eq:resheuristic}, the AT condition can also be expressed in terms of a spectral condition. To see this, let us neglect the rank-one perturbation $A$ and the correlations between $\Lambda$ and $G$, which should be weak at high temperature. Setting $G = \sqrt{t} \widetilde{G}$ for a GOE matrix $\widetilde{G}$, we are evaluating
\[M(E) = \left(\Lambda - \sqrt{t}\widetilde{G} - E \right)^{-1}\]
at a special energy $E_0 = -t(1-q)$. From random matrix theory we expect
\[M_{ii}(E) = \frac{1}{\Lambda_{ii} - E - t S(E)}\]
with
\[S(E) = \frac{1}{N} \sum_i \frac{1}{\Lambda_{ii} - E - t S(E)}.\]
Here, $E$ can be real as long as it is outside of the spectrum. Now notice that
\begin{align*} S^\prime (E) &= (1 + t S^\prime(E)) \frac{1}{N} \sum_i \frac{1}{(\Lambda_{ii} - E - t S(E))^2}\\
&= (1 + t S^\prime(E)) \frac{1}{N} \sum_i (M_{ii}(E))^2.
\end{align*}
If we plug in  $E_0 = -t(1-q)$, this calculation says that
\[S^\prime(E_0) = (1 + t S^\prime(E_0)) \frac{1}{N}  \sum_{i=1}^N (1-m_i^2)^2 \approx  (1 + t S^\prime(E_0)) \ee \sech^{4}(h +   \sqrt{tq} Z),  \]
so that
\[S^\prime(E_0) = \frac{  \ee \sech^{4}(h +   \sqrt{tq} z)}{1 -  t \ee \sech^{4}(h +  \sqrt{tq} z)}.  \]
In particular, $S^\prime(E_0)$ is finite precisely under the AT condition. Since $S$ is supposed to be analytic everywhere except the spectral edge, this fits in nicely with $E_0$ being outside the spectrum under the AT condition.
\end{remark}

Let us conclude this introduction with some comments about how to extend our results to mixed $p$-spin models. To this end, let $ H_N^{(p)}:\{-1,1\}^N\to \mathbb{R}$ be defined by
		\[ H_N^{(p)}(\sigma) = h + \beta \sum_{p=2}^\infty \frac{\beta_p \sqrt{p!}}{N^{(p-1)/2}}\sum_{ |A|=p}  g_{A} \prod_{i\in A} \sigma_i\]
for i.i.d. standard Gaussian random variables $ (g_A)_{A\subset \{1,\dots,N\}}$ and a sequence $(\beta_p)_{p\geq 2}$ ensuring that $ \xi (s) := \beta^2\sum_{p=2}^\infty \beta_p^2s^p<\infty$ for all $s\in[0;1]$. The function $\xi$ characterizes the model in the sense that
		\[\mathbb{E}\, (H_N^{(p)}(\sigma^1)-h)(H_N^{(p)}(\sigma^2)-h) =  \xi\big( R_{1,2}(\sigma^1,\sigma^2)\big).  \]
Analogously to Theorem \ref{thm:main}, one can prove that for $ \beta\geq0 $ sufficiently small and $ \beta_p =\beta_0^p$ for some $\beta_0\geq 0$ sufficiently small, there exists a constant $C = C_{\beta, \beta_0}>0$ such that
		\begin{equation}\label{eq:pspinbnds}
		\begin{split}
		&\mathbb{E} \bigg[m_i - \tanh\bigg( h +  \beta\sum_{p=2}^\infty \frac{\beta_p \sqrt{p!}}{N^{(p-1)/2}}\sum_{i\in A,  |A|=p}  g_{A} \Big\langle\prod_{k\in A, k\neq i}  \sigma_k  \Big\rangle^{(i)}  \bigg)\bigg]^2\leq \frac{C}N,  \\
		&\mathbb{E} \bigg[ m_{ij} -  \sech^2\bigg( h +  \beta\sum_{p=2}^\infty \frac{\beta_p \sqrt{p!}}{N^{(p-1)/2}}\sum_{i\in A,  |A|=p}  g_{A} \Big\langle\prod_{k\in A, k\neq i}  \sigma_k  \Big\rangle^{(i)}  \bigg)  \\
		&\hspace{3.5cm} \times  \beta\sum_{p=2}^\infty \frac{\beta_p \sqrt{p!}}{N^{(p-1)/2}}\sum_{i\in A,  |A|=p}  g_{A}    \Big\langle \sigma_j; \prod_{k\in A, k\neq i}  \sigma_k   \Big\rangle^{(i)} \bigg]^2\leq \frac{C}{N^{1+\epsilon}} 
		\end{split}
		\end{equation}
for all $\epsilon>0$ sufficiently small and $i\neq j$. Here, we denote
		\[  \Big\langle \sigma_j; \prod_{k\in A, k\neq i}  \sigma_k   \Big\rangle^{(i)} = \Big\langle \sigma_j \prod_{k\in A, k\neq i}  \sigma_k   \Big\rangle^{(i)}- \langle\sigma_j\rangle^{(i)}\Big\langle  \prod_{k\in A, k\neq i}  \sigma_k   \Big\rangle^{(i)} \]
and all Gibbs expectations are taken with respect to the Gibbs measure induced by $H_N^{(p)}$. Since the methods to prove Theorem \ref{thm:main} can be adapted in a straight-forward way to prove the bounds in \eqref{eq:pspinbnds}, we focus in this paper exclusively on the analysis of the $2$-spin model with Hamiltonian $H_N$ defined in \eqref{eq:HN}.

Finally, let us remark that also the heuristics in \eqref{eq:heuristic} can be generalized to the $p$-spin models. Indeed, let us assume appropriate decay of correlations so that we can factorize
		\[  \Big\langle\prod_{k\in A, k\neq i}  \sigma_k  \Big\rangle^{(i)} \approx \prod_{k\in A, k\neq i} m_k^{(i)}. \] 
Writing $ A = \{ j_1,j_2,\dots, j_p\}$, this can be made rigorous by using the identity
		\[\begin{split}
		& \Big\langle\prod_{k\in A, k\neq i}  \sigma_k  \Big\rangle^{(i)}  - \prod_{k\in A, k\neq i} m_k^{(i)} \\
		& \hspace{0.5cm}=  \Big\langle \sigma_{j_1}; \prod_{\substack{k\in A,\\ k\neq i,j_1}}  \sigma_k  \Big\rangle^{(i)} + \Big\langle \sigma_{j_2}; \prod_{\substack{ k\in A,\\ k\neq i,j_1, j_2}}  \sigma_k  \Big\rangle^{(i)}m_{j_1}^{(i)}  +\ldots+  m_{j_{p-1}j_p}^{(i)}\prod_{\substack{k\in A, \\k\neq i, j_{p-1},j_p}} m_k^{(i)}   
		\end{split}  \]		
and adapting the methods presented below to show that the correlation functions on the right hand side are small in the limit $N\to\infty$. By \eqref{eq:pspinbnds} and in analogy to Prop. \ref{prop:concentration}, we then expect that $ q_N = N^{-1}\sum_{k=1}^Nm_k^2 \approx \mathbb{E}\,q_N$ concentrates and converges as $N\to\infty$ to a solution $q\in[0;1]$ of the self-consistent equation
		\[q  = \mathbb{E} \tanh^2( h + \sqrt{\xi'(q)}Z).\]
Here, $Z\sim \mathcal{N}(0,1)$ denotes a standard Gaussian. Assuming similarly that
		\[    \Big\langle \sigma_j; \prod_{k\in A, k\neq i}  \sigma_k   \Big\rangle^{(i)} \approx \sum_{k \in A, k\neq i}  m_{jk} \prod_{l \in A , l \neq i,k } m_l,  \] 
we may follow the heuristics of \eqref{eq:heuristic} and expect that
		\[\begin{split}
		 \mathbb{E}\,  m_{ij}^2  &\approx \frac1N \frac{\mathbb{E} \sech^4(h +  \sqrt{\xi'(q)} Z\big)\, \mathbb{E}\, \xi''(q)\sech^4(h +  \sqrt{\xi'(q)} Z\big)}{1-\mathbb{E}\, \xi''(q)\sech^4(h +  \sqrt{\xi'(q)} Z\big) }.  
		 \end{split}\] 
In particular, this can only hold true under the condition
		\[\mathbb{E}\, \xi''(q)\sech^4(h +  \sqrt{\xi'(q)} Z\big) <1, \]
which appears to be consistent with the generalized AT condition that is conjectured in \cite[Eq. (1.8) \& Eq. (1.9)]{JagTob} (assuming that, at sufficiently high temperature, the self-consistent equation $q = \mathbb{E} \tanh^2( h + \sqrt{\xi'(q)}Z)$ has a unique fixed point).

The paper is structured as follows. In the following Section \ref{sec:notation}, we introduce our notation. In Section \ref{sec:corrdec}, we establish suitable decay bounds on the two and three point correlation functions. In Sections \ref{sec:hTAP} and \ref{sec:TAP}, we prove the TAP equations in the sense of Theorem \ref{thm:main} and Corollary \ref{cor:main}.
Finally, in Section \ref{sec:conc}, we prove Propositions \ref{prop:concentration} and \ref{prop:mij2norm}.

\medskip

\noindent\textbf{Acknowledgements.} The work of P.\,S. is supported by the DFG grant SO 1724/1-1. Part of this work was written when A.A. was under the sponsorship of a Harvard University GSAS MGSTTRF. The research of H.-T. Y. is partially supported by NSF grant DMS-1855509 and a Simons Investigator award.

\section{Notation}\label{sec:notation}

In the following, we will need to consider expectations of observables conditionally on a given number of spins. To this end, it is useful to set up the following notation. Let $ A=  \{j_1, j_2,\dots, j_k\}  \subset \{1,\dots, N\}$, let $B\subset \{1,\dots, N\}$ be disjoint from $A$ with $|B|=l$ and let $\mathbf{\tau} =  (\tau_{j_1}, \dots, \tau_{j_k})\in \{-1,1\}^k$ be a fixed $j$-particle configuration. Then, we define the reduced Hamiltonian $ H_N^{ [A,B]} \equiv H_{N,(\tau_{j_1}, \dots, \tau_{j_k})}^{[A,B]}:  \{-1,1\}^{N-k-l}\to\mathbb{R}$ by
		\[
		H_N^{[A,B]}(\mathbf{\sigma}) = H_{N }^{[A,B]}(\sigma_{i_1}, \dots, \sigma_{i_{N-k-l}}) = \sum_{\substack{ 1\leq i< j\leq N:\\ i,j \not \in A \cup B  }} g_{ij} \sigma_{i}\sigma_j + \sum_{\substack{ 1\leq i\leq N: \\i\not \in A \cup B }} \Big( h + \sum_{j\in A } g_{ij}\tau_j   \Big)  \sigma_i.
		\]
$ H_N^{[A,B]}(\mathbf{\sigma})$ plays the role of the energy of the system, conditionally on the spins $ \sigma_j $ for $j\in A$ such that $\sigma_j=\tau_j$ and after the particles $\sigma_j$ for $j\in B$ have been removed from the system. For disjoint subsets $ A, B\subset \{1,\dots, N\}$, we then denote by $ \langle \cdot \rangle^{[A,B]}$ the Gibbs measure induced by the reduced Hamiltonian $ H_N^{[A,B]}$. We abbreviate $ \langle \cdot \rangle^{[A]} \equiv \langle\cdot \rangle^{[A, \emptyset]}$, $ \langle\cdot \rangle^{ (B)} \equiv  \langle\cdot \rangle^{[\emptyset, B]}$ as well as
$ \langle\cdot \rangle\equiv  \langle\cdot \rangle^{[\emptyset, \emptyset]}$. In particular, $ \langle\cdot\rangle$ denotes the usual Gibbs measure induced by $H_N = H_N^{[\emptyset,\emptyset]}$. By slight abuse of notation, if $A=\{i\}$ is a set of only one element, we write for simplicity \[ \langle\cdot\rangle^{[i]}:= \langle\cdot\rangle^{[\{i\}]}, \hspace{1cm}\langle\cdot\rangle^{(i)}:= \langle\cdot\rangle^{(\{i\})}.\]

For an observable $f$, notice that $ \langle f \rangle^{[A]}$ is equal to the conditional expectation of $f$, given the spins $\sigma_j$ for $j\in A$. Observables of particular interest will be the magnetizations $m_i^{[A]}$, the two point functions $ m_{ij}^{[A]} $ and the three point functions $m_{ijk}^{[A]}$, defined by
		\[\begin{split}m_i^{[A]} &=  \langle \sigma_i \rangle^{[A]}, \hspace{0.5cm} m_{ij}^{[A]} = \langle \sigma_i \sigma_j \rangle^{[A]} - \langle\sigma_i\rangle^{[A]}\langle \sigma_j\rangle^{[A]}, \\
		m_{ijk}^{[A]}& = \big\langle \big(\sigma_i- \langle\sigma_i\rangle^{[A]}\big)  \big( \sigma_j - \langle\sigma_j\rangle^{[A]}\big)\big( \sigma_k - \langle\sigma_k\rangle^{[A]}\big)\big\rangle^{[A]}. \end{split}\]
 If $ A=\emptyset$, we simply write $ m_i$, $m_{ij}$ and $m_{ijk}$, respectively.
		
Given disjoint subsets $A, B\subset \{1,\dots, N\}$, an index $i\in A$ and an observable $f$, we introduce furthermore the notation
		\[\delta_i \langle f \rangle^{[A,B]} = \frac12  \sum_{\sigma_i=\pm1 } \sigma_i \langle f \rangle ^{[A,B]} (\sigma_i) , \hspace{0.5cm} \eps_i \langle f \rangle^{[A]} = \frac12  \sum_{\sigma_i=\pm1 } \langle f \rangle^{[A,B]} (\sigma_i).\]

Finally, we denote by $C$ generic constants that may vary from line to line and that are independent of all parameters, unless specified otherwise. If a constant depends on a parameter, say $\epsilon$, we denote this typically by a subscript, i.e. $C_\epsilon$.

\section{Bounds on Correlation Functions}\label{sec:corrdec}

In this section, we will bound the two and three point functions, based on the key identity
		\begin{equation}\label{eq:condid} \begin{split}
		m_{ij}^{[A]}&=   \Big[1-\big( m_i^{[A]}\big)^2\Big] \delta_i m_j^{[A\cup \{i\} ]}.
		\end{split} \end{equation}
Differentiating~\eqref{eq:condid} with respect to the external field in direction of $\sigma_k$, we also get
		\begin{equation}\label{eq:condid2} m_{ijk}^{[A]} =   \Big[1-\big( m_i^{[A]}\big)^2\Big] \delta_i m_{jk}^{[A\cup \{i\} ]} - 2m_i^{[A]} m_{ik}^{[A]}\delta_i m_{j}^{[A\cup \{i\} ]}.\end{equation}
Here, $ A\subset \{1,\dots, N\} $ and $i,j,k \not \in  A $. Equation \eqref{eq:condid} is a simple consequence of the fact that the spins take values in $\{-1,1\}$ and the identities
		\[\begin{split}
		\langle\sigma_j\rangle^{[A]} =  m_j^{[A\cup\{i\}] }(\sigma_i=1) \big\langle\mathbf{1}_{\sigma_i=1}\big\rangle^{[A]} + m_j^{[A\cup\{i\}] }(\sigma_i=-1) \big\langle\mathbf{1}_{\sigma_i=-1}\big\rangle^{[A]} , \\
		\langle \sigma_i \sigma_j \rangle ^{[A]} = m_j^{[A\cup\{i\}] }(\sigma_i=1) \big\langle\mathbf{1}_{\sigma_i=1}\big\rangle^{[A]} - m_j^{[A\cup\{i\}] }(\sigma_i=-1) \big\langle\mathbf{1}_{\sigma_i=-1}\big\rangle^{[A]}.
		\end{split}\]

Let us consider first the two point functions. A simple idea to control the two point functions is to expand the identity \eqref{eq:condid} dynamically in the randomness $(g_{ik})_{k  \not \in A}$. More precisely, we can view the $ (g_{ik})_{k \not \in A}$ in $ H_N^{[A\cup\{i\}]}$ as Brownian motions at time $t$ and speed $1/N$  to rewrite the difference $ \delta_i m_j^{[A\cup\{i\} ]}$ in \eqref{eq:condid} through It\^o's lemma as
		\begin{equation}\label{eq:dynid} \begin{split}
		 \delta_i m_j^{[A\cup\{i\} ]}&=   \sum_{\substack{ k \not \in A  }}\int_0^t  \eps_i\, m_{kj}^{[A\cup\{i\}  ]}   (s)\,dg_{ik}(s) -   \sum_{\substack{ k \not \in A } }\int_0^t  \delta_i \Big(m_k^{[A\cup\{i\}]} m_{kj}^{[A\cup\{i\}  ]}\Big)(s) \frac{ds}N \\
		\end{split}
		\end{equation}	
Here and throughout this paper, we abbreviate $ \langle f \rangle^{[A\cup\{i\}]}(s) = \langle f \rangle^{[A\cup\{i\}]} \big( (g_{il}(s))_{l\not \in A}\big)$ for any observable $f$.

\begin{lemma}\label{lm:L2bndmij} Let  $0\leq  t < \log 2$, let $A\subset \{1,\dots, N\}$ and choose $\epsilon>0$ sufficiently small. Then, for some $C_{t,\epsilon}>0$, independent of $N$ and $A\subset \{1,\dots,N\}$, we have that
		\[ \sup_{\sigma \in \{-1,1\}^{|A|}}\mathbb{E}\, \big|m_{ij}^{[A]}\big|^{2+\epsilon} \leq  \frac{C_{t,\epsilon} }{N^{1+\epsilon/2}}  \]
for all $i\neq j$ with $i,j\not \in A$.
\end{lemma}
\begin{proof} Let $ A\subset \{1,\dots, N\}$ be arbitrary. By \eqref{eq:condid}, we have that
		\[ \mathbb{E} \big|m_{ij}^{[A]}\big|^{2+\epsilon} \leq \mathbb{E} \big|\delta_i m_j^{[A\cup\{i\}]}\big|^{2+\epsilon}, \]
so let us bound the right hand side. It\^o's Lemma and~\eqref{eq:dynid} imply
		\[\begin{split}
\mathbb{E} \big| \delta_i m_j^{[A\cup\{i\}]}\big|^{2+\epsilon}(t)   &\le   (1+\epsilon/2)(1+\epsilon)  \sum_{\substack{ k \not \in A }}\int_0^t \mathbb{E}  \Big|\delta_i m_j^{[A\cup\{i\}]}\Big|^{\epsilon}(s) \Big| \eps_i m_{kj}^{[A\cup\{i\}  ]} \Big|^2 (s)\, \frac{ds}N\\
		&+ (2+\epsilon)   \sum_{\substack{ k \not \in A } }\int_0^t   \mathbb{E} \, \Big|\delta_i m_j^{[A\cup\{i\}]} \Big|^{1+\epsilon} \Big|(s) \delta_i \Big(m_k^{[A\cup\{i\}]} m_{kj}^{[A\cup\{i\}  ]}\Big) \Big|(s) \frac{ds}N \\
		&\leq \frac{C_{t,\epsilon}}{N^{1+\epsilon/2}} +  (1+3 \epsilon)\int_0^t \mathbb{E}\big|\delta_i m_j^{[A\cup\{i\}]}\big|^{2+\epsilon}(s)\, ds \\
		&+ \sup_{\substack{ k\not \in A\cup\{j\}, \\\sigma_i=\pm1 }} \int_0^t (1+\epsilon)\,\mathbb{E}\Big(1+\big|m_k^{[A\cup\{i\}]}\big|^{2+\epsilon}\Big) \big|m_{kj}^{[A\cup\{i\}]}(\sigma_i)\big|^{2+\epsilon} (s)\,ds.
		\end{split} \]
Here we used Young's inequality, the smallness of $\epsilon$, and the trivial bounds for the case $k=j$. Inserting \eqref{eq:condid}, we obtain
		\[\begin{split}
		\Big(1+\big|m_k^{[A\cup\{i\}]}\big|^{2+\epsilon}\Big) \big|m_{kj}^{[A\cup\{i\}]}(\sigma_i)\big|^{2+\epsilon}  &\leq \Big[1-\big|m_k^{[A\cup\{i\}]}\big|^2\Big]^{1+\epsilon}  \big|\delta_k m_{j}^{[A\cup\{i,k\}]}(\sigma_i)\big|^{2+\epsilon}\\
		&\leq \big|\delta_k m_{j}^{[A\cup\{i,k\}]}(\sigma_i)\big|^{2+\epsilon},
		\end{split}\]
so that
	\[\begin{split}
		\mathbb{E} \big|\delta_i m_j^{[A\cup\{i\}]}\big|^{2+\epsilon}(t)  &\leq \frac{C_{t,\epsilon}}{N^{1+\epsilon/2}} + (1+3\epsilon) \int_0^t \mathbb{E}\big|\delta_i m_j^{[A\cup\{i\}]}\big|^{2+\epsilon}(s)\, ds\\
		&\hspace{0.5cm} + (1+ \epsilon)\sup_{\substack{ k\not \in A\cup\{j\}, \\\sigma_i=\pm1 }} \int_0^t \mathbb{E} \big|\delta_k m_{j}^{[A\cup\{i,k\}]}(\sigma_i)\big|^{2+\epsilon} (s)\,ds.
		\end{split} \]			
Combining Gronwall's inequality with integration by parts then shows that, uniformly in $\sigma\in \{-1,1\}^{|A|}$, we have
	\[\mathbb{E} \big|\delta_i m_j^{[A\cup\{i\}]}\big|^{2+\epsilon}(t)\leq  \frac{C_{t,\epsilon}}{N^{1+\epsilon/2}} + (1+\epsilon)\sup_{\substack{ k\not \in A\cup\{j\}, \\\sigma_i=\pm1 }} \int_0^t e^{(1+3\epsilon)(t-s)}\,\mathbb{E} \big|\delta_k m_{j}^{[A\cup\{i,k\}]}(\sigma_i)\big|^{2+\epsilon} (s)\,ds.\]
		
Now, since $A\subset \{1,\dots,N\}$ was arbitrary, we may iterate the last bound by viewing the rows of $ (g_{ij})_{1\leq i<j\leq N}$ successively as Brownian motions at time $t$ and of speed $1/N$. This way, we obtain that
		\[\begin{split}
		&\mathbb{E} \big|\delta_i m_j^{[ A\cup\{i\}]}\big|^{2+\epsilon} \leq  \frac{C_{t,\epsilon}}{N^{1+\epsilon/2}} + (1+\epsilon) \!\sup_{\substack{ k_1 \neq i, j; \\\sigma_i=\pm1 }} \int_0^t e^{(1+3\epsilon)(t-s_1)}\,\mathbb{E} \big|\delta_{k_1} m_{j}^{[A\cup\{i,k_1\}]}(\sigma_i)\big|^{2+\epsilon}(s_1)ds_1\\
		&\leq  \frac{C_{t,\epsilon}}{N^{1+\epsilon/2}} \bigg[ 1+ \frac{(1+\epsilon)(e^{(1+3\epsilon)t}-1)}{(1+3\epsilon)}\bigg] \\
		&\hspace{0.5cm}+(1+\epsilon)^2 \!\! \!\! \sup_{\substack{ k_1 \neq i,j; \\k_2\neq i,j,k_1; \\\sigma_i,\sigma_{k_1}=\pm1 }}\!\! \int_0^t \int_0^t  e^{(1+3\epsilon)(2t-s_1-s_2)}\,\mathbb{E} \big| \delta_{k_2}m_j^{[A\cup\{i,k_1,k_2\}]}(\sigma_i, \sigma_{k_1})\big|^{2+\epsilon} (s_1;s_2)\,ds_1ds_2\\
		&\leq \frac{C_{t,\epsilon}}{N^{1+\epsilon/2}} \Big[ 1+ \frac{(1+\epsilon)(e^{(1+3\epsilon)t}-1)}{(1+3\epsilon)} +\dots + \frac{(1+\epsilon)^{n-1}(e^{(1+3\epsilon)t}-1)^{n-1}}{(1+3\epsilon)^{n-1}}\Big] \\
		&\hspace{0.5cm} +(1+\epsilon)^{n-1} \!\!\sup_{\substack{ k_1 \neq i,j; \\k_2\neq i,j,k_1;  }} \ldots \sup_{\substack{ k_n \neq i,j, k_1,\dots, k_{n-1}; \\\sigma_i,\sigma_{k_1},\dots, \sigma_{k_{n-1}} =\pm1 }} \int_0^t \int_0^t \ldots \int_0^t  e^{(1+3\epsilon)\sum_{m=1}^n(t-s_m)} \,\\
		&\hspace{1.5cm}\times \mathbb{E} \big|\delta_{k_n}m_j^{[A\cup\{i,k_1,\dots, k_{n-1}\}]}(\sigma_i, \sigma_{k_1},\dots, \sigma_{k_{n-1}})\big|^{2+\epsilon} (s_1;s_2; \dots;s_n)\,ds_1ds_2\ldots ds_n
		\end{split}\]
for every $n\leq N-|A|$. Here, we used similarly as above the notation
		\[\langle f \rangle^{[\{i,k_1,\dots, k_{n-1}\}]}(s_1;s_2; \dots;s_n) \equiv \langle f \rangle^{[\{i,k_1,\dots, k_{n-1}\}]}\big( g_{i\bullet} (s_1); g_{k_1\bullet}(s_2); \dots; g_{k_{n-1}\bullet}(s_n)\big) \]
for an observable $f$. In particular, the above estimate implies for $  t < \log 2$ and $\epsilon>0$ sufficiently small that, uniformly in $\sigma\in \{-1,1\}^{|A|}$, we have
		\[ \mathbb{E} \,\big|m_{ij}^{[A]}\big|^{2+\epsilon} 
		\leq  \frac{C_{t,\epsilon}}{N^{1+\epsilon/2}}. \]
\end{proof}

{\noindent\emph{Remarks:}}
\begin{itemize}
\item[1)]  By optimizing the Gronwall argument from the previous proof, one can improve the lemma to hold for all times $t\geq 0$ that satisfy
		\[0\leq t< \max_{x\in [2;\infty)} x\log \bigg[ 1+ x^{-1}\Big(\frac13+\frac x3\Big)^{-1}\Big(\frac23+\frac2{3x}\Big)^{-2}\bigg] \approx 0.83.\]
\item[2)] The bound provided in Lemma~\ref{lm:L2bndmij} is clearly uniform in time. More precisely,  we have
		\begin{equation}\label{eq:L2bnddimjA}
		\begin{split}
		& \sup_{s_{ij}\in [0;t], \,1\leq i<j\leq N}\mathbb{E} \,\big|m_{ij}^{[A]}\big|^{2+\epsilon}\big( (g_{ij}(s_{ij}))_{1\leq i<j\leq N}\big)\\
		& \leq \sup_{s_{ij}\in [0;t], \,1\leq i<j\leq N}\mathbb{E} \, \big|\delta_{i}m_j^{[A\cup\{i\}]}\big|^{2+\epsilon} \big( (g_{ij}(s_{ij}))_{1\leq i<j\leq N}\big)\leq  \frac{C_{t,\epsilon} }{N^{1+\epsilon/2}}
		\end{split}
		\end{equation}
		uniformly in $\sigma\in \{-1,1\}^{|A|}$ and $t<\log 2$.
\item[3)] The estimate for $ m_{ij}^{[A]}$ in $ L^{2+\epsilon}(\mathbb{P})$, rather than $L^2(\mathbb{P})$, is required to obtain an estimate for the three point functions $m_{ijk}$ in $L^2(\mathbb{P})$, see Lemma \ref{lm:L2bndmijk} below. The previous proof can also be adapted to bound higher moments of $|m_{ij}^{[A]}|$. If one applies It\^{o}'s lemma to the $L^p$-norm and chooses appropriate new exponents in Young's inequality, the same argument proves that for any $p\in [2;\infty)$,  there exists some sufficiently small $ t=t_p>0$ with
		\[ \sup_{\sigma \in \{-1,1\}^{|A|}}\mathbb{E}\, \big|m_{ij}^{[A]}\big|^{p} \leq  \frac{C_{t,p } }{N^{p/2}}. \]
	Similar remarks apply to the remaining arguments in this paper. In particular, adapting the proof of Lemma \ref{lm:hTAP1} below yields the validity of the TAP equations \eqref{eq:hTAP1} in $L^p(\mathbb{P})$. If we choose $p\geq 2$ sufficiently large, this also shows that the TAP equations \ref{eq:hTAP1} hold simultaneously for all $m_i$ with high probability (however, only for sufficiently small times $t=t_p>0$).
\end{itemize}

In the next section, we will also need rough bounds on the three point functions.
\begin{lemma}\label{lm:L2bndmijk} Let  $0\leq  t < \log 2$, let $A\subset \{1,\dots, N\}$ and choose $\epsilon>0$ sufficiently small. Then, for some $C_{\epsilon}>0$, independent of $N$, $t$ and $A\subset \{1,\dots,N\}$, we have that
		\[ \sup_{\sigma \in \{-1,1\}^{|A|}}\mathbb{E}\, \big|m_{ijk}^{[A]}\big|^{2} \leq  \frac{C_{t,\epsilon}}{N^{1+\epsilon/2}} \]
for all $i\neq j, i\neq k, j\neq k$ and $i,j,k\not \in A$.
\end{lemma}
\begin{proof}
Lemma \eqref{lm:L2bndmij} and the Cauchy-Schwarz inequality combine to show that
		\begin{equation}\label{eq:bndmijprod}  \begin{split} \big\|  m_{ik}^{[A]} \delta_i m_j^{[A\cup\{i\}]}\big\|_2^2 &\leq  \big\|m_{ik}^{[A]}\big\|_4^2 \big\|\delta_i m_j^{[A\cup\{i\}]}\big\|_4^2\\
		&\leq \big\|m_{ik}^{[A]}\big\|_{2+\epsilon}^{1+\epsilon/2} \big\|\delta_i m_j^{[A\cup\{i\}]}\big\|_{2+\epsilon}^{1+\epsilon/2} \leq \frac{C_{t,\epsilon}}{N^{1+\epsilon/2}}. \end{split}\end{equation}
By the identity \eqref{eq:condid2}, it is therefore enough to control $ \delta_{i}m_{jk}^{[A\cup\{i\}]}$. Differentiating the identity \eqref{eq:dynid} with respect to the external field in the direction of $\sigma_k$, we find that
		\[ \begin{split}
		 \delta_i m_{jk}^{[A\cup\{i\} ]} &=  \sum_{\substack{ l \not \in A }}\int_0^t  \eps_i\, m_{jkl}^{[A\cup\{i\}  ]}   (s)\,dg_{il}(s)  -  \sum_{\substack{ l \not \in A} }\int_0^t  \delta_i \Big(m_{kl}^{[A\cup\{i\}]} m_{jl}^{[A\cup\{i\}  ]}\Big)(s) \frac{ds}N \\
 &- \sum_{\substack{ l \not \in A} }\int_0^t  \delta_i \Big(m_{l}^{[A\cup\{i\}]} m_{jkl}^{[A\cup\{i\}  ]}\Big)(s) \frac{ds}N.
		\end{split}
		\]	
We proceed as in Lemma \ref{lm:L2bndmij}, using the It\^{o} isometry followed by Young's inequality. If we also apply the trivial bound to the summands with $l \in \{j, k\}$ and insert the bounds of Lemma \ref{lm:L2bndmij} for the two-point functions, we conclude that
		\[\begin{split} \mathbb{E} \big|\delta_i m_{jk}^{[A\cup\{i\} ]}\big|^2 & \leq \frac{C_{t,\epsilon}}{N^{1+\epsilon/2}} +   \sum_{\substack{ l \not \in A \cup\{j,k\} }}\int_0^t \mathbb{E} \big| \eps_i\, m_{jkl}^{[A\cup\{i\}  ]}  \big|^2 (s)\frac{ds}N\\
		&\quad -2 \sum_{\substack{ l \not \in A \cup\{j,k\}} }\int_0^t \mathbb{E} \big(\delta_i m_{jk}^{[A\cup\{i\} ]}\big) (s)\delta_i \Big(m_{l}^{[A\cup\{i\}]} m_{jkl}^{[A\cup\{i\}  ]}\Big)(s) \frac{ds}N\\
		& \leq \frac{C_{t,\epsilon}}{N^{1+\epsilon/2}} + \int_0^t \mathbb{E} \big|\delta_i m_{jk}^{[A\cup\{i\} ]}\big|^2 (s) \,ds\\
		&\quad +  \sup_{\substack{ l\not \in A\cup\{j,k\}, \\\sigma_i=\pm1 }}\int_0^t \mathbb{E} \Big(1+ \big| m_{l}^{[A\cup\{i\}]}\big|^2\Big)\big|   m_{jkl}^{[A\cup\{i\}  ]}  \big|^2 (s)\,ds ,\\
		\end{split}\]
uniformly in $A\subset \{1,\dots,N\}$. Using once more the identity \eqref{eq:condid2} together with the results of Lemma \ref{lm:L2bndmij} and the remarks following its proof, we have that
		\[ \sup_{s\in[0;t]}  \mathbb{E}   \Big| \big|   m_{jkl}^{[A\cup\{i\}  ]}  \big|^2 (s) -  \Big[1-\big| m_l^{[A\cup\{i\}]}\big|^2\Big]^2 \big|\delta_l m_{jk}^{[A\cup \{i,l\} ]}\big|^2(s) \Big|\leq \frac{C_{t,\epsilon}}{N^{1+\epsilon/2}} \]
and hence
		\[\begin{split} \mathbb{E} \big|\delta_i m_{jk}^{[A\cup\{i\} ]}\big|^2 & \leq \frac{C_{t,\epsilon}}{N^{1+\epsilon/2}} + \int_0^t \mathbb{E} \big|\delta_i m_{jk}^{[A\cup\{i\} ]}\big|^2 (s) \,ds  +  \sup_{\substack{ l\not \in A\cup\{j,k\}, \\\sigma_i=\pm1 }}\int_0^t \mathbb{E}\big|\delta_l m_{jk}^{[A\cup \{i,l\} ]}\big|^2(s)\,ds.
		\end{split}\]
Gronwall's Lemma implies that
		\[\begin{split} \mathbb{E} \big|\delta_i m_{jk}^{[A\cup\{i\} ]}\big|^2 & \leq \frac{C_{t,\epsilon}}{N^{1+\epsilon/2}}   +  \sup_{\substack{ l\not \in A\cup\{j,k\}, \\\sigma_i=\pm1 }}\int_0^t e^{t-s}\,\mathbb{E}\big|\delta_l m_{jk}^{[A\cup \{i,l\} ]}\big|^2(s)\,ds
		\end{split}\]
and by iterating this estimate $N-|A|$ times, as in the proof of Lemma \ref{lm:L2bndmij}, we find
		\[ \mathbb{E} \big|\delta_i m_{jk}^{[A\cup\{i\} ]}\big|^2 \leq \frac{C_{t,\epsilon}}{N^{1+\epsilon/2}}\]
for $t<\log 2$, uniformly in $A\subset \{1,\dots, N\}$. Together with \eqref{eq:bndmijprod}, this proves the claim.
\end{proof}
\noindent \emph{Remark:}
\begin{itemize}
\item[1)] Viewing the $(g_{ij})_{1\leq i< j\leq N}$ dynamically as in the previous proof, the same arguments imply also that, uniformly in $\sigma\in \{-1,1\}^{|A|}$, we have for $t < \log 2$ that
		\[
		\begin{split}
		&\sup_{s_{ij}\in [0;t], \,1\leq i<j\leq N }\mathbb{E} \,\big|m_{ijk}^{[A]}\big|^{2 }\big((g_{ij}(s_{ij}))_{1\leq i<j\leq N}\big) \\
		&\hspace{1cm}\leq \sup_{ s_{ij}\in [0;t], \,1\leq i<j\leq N }\mathbb{E} \, \big|\delta_{i}m_{jk}^{[A\cup\{i\}]}\big|^{2 }\big((g_{ij}(s_{ij}))_{1\leq i<j\leq N}\big)+ \frac{C_{t,\epsilon}}{N^{1+\epsilon/2}}\leq  \frac{C_{t,\epsilon}}{N^{1+\epsilon/2}}.
		\end{split}
		\]
\end{itemize}


\section{Proof of the Hierarchical TAP Equations}\label{sec:hTAP}

Using the bounds on the size of the correlation functions, we are now ready to prove the validity of the hierarchical TAP equations for the one and two point functions in the sense of $ L^2(\mathbb{P})$. This will prove in particular our main result Theorem \ref{thm:main}.

\begin{lemma}\label{lm:hTAP1} Let  $0\leq  t < \log 2$. Then, for some $C=C_t>0$ independent of $N$, we have
		\[ \mathbb{E} \Big[m_i - \tanh\Big( h + \sum_{j\neq i} g_{ij}m_j^{(i)}\Big)\Big]^2\leq \frac{C}N.\]
\end{lemma}
\begin{proof}
By the Lipschitz continuity of $\tanh(\cdot)$, the claim follows if we show that
		\[ \mathbb{E} \Big[\tanh^{-1}(m_i) - \Big( h + \sum_{j\neq i} g_{ij}m_j^{(i)}\Big)\Big]^2\leq \frac{C}N.\]
We view the $(g_{ij})_{1\leq j\leq N}$ dynamically as Brownian motions at time $t$ and of speed $1/N$ so that a straight forward application of It\^o's Lemma implies that
		\[\begin{split}
		\tanh^{-1}(m_i) - \Big( h + \sum_{j\neq i} g_{ij}m_j^{(i)}\Big) & = \sum_{j\neq i } \int_0^t \bigg(m_j- m_j^{(i)} - \frac{m_im_{ij}}{1-m_i^2} \bigg) (s)\,dg_{ij}(s)  \\
		&\hspace{0.5cm} -\int_0^t \sum_{j\neq i} \bigg(\frac{m_j m_{ij}}{1-m_i^2}-\frac{m_i m_{ij}^2}{1-m_i^2}-\frac{m_i^3m_{ij}^2}{(1-m_i^2)^2} \bigg)(s) \frac{ds}N.
		\end{split}\]
Recalling that $ m_{ij} /(1-m_i^2) =\delta_i m_j^{[i]}$, we use $|m_i|\leq 1$, $ |\delta_i m_j^{[i]}|\leq 2$ to conclude that
		\[\begin{split}
		\bigg\| \int_0^t \sum_{j\neq i} \bigg(\frac{m_j m_{ij}}{1-m_i^2}-\frac{m_i m_{ij}^2}{1-m_i^2}-\frac{m_i^3m_{ij}^2}{(1-m_i^2)^2} \bigg)(s) \frac{ds}N\bigg\|_2
		\leq C\sum_{j\neq i}\int_0^t \|\delta_i m_j^{[i]}(s)\|_2 \,\frac{ds}N
		\end{split}\]
By the observation \eqref{eq:L2bnddimjA}	after the proof of Lemma \ref{lm:L2bndmij}, this implies that
		\begin{equation}\label{eq:L2TAP1driftbnd} \mathbb{E}\bigg[ \int_0^t \sum_{j\neq i} \bigg(\frac{m_j m_{ij}}{1-m_i^2}-\frac{m_i m_{ij}^2}{1-m_i^2}-\frac{m_i^3m_{ij}^2}{(1-m_i^2)^2} \bigg)(s) \frac{ds}N\bigg]^2
\leq \frac{C}N.\end{equation}
Similarly, it follows that
		\begin{equation} \label{eq:L2TAP1martingale1bnd} \mathbb{E}\bigg[ \int_0^t \sum_{j\neq i }  \frac{m_im_{ij}}{1-m_i^2}   (s)\,dg_{ij}(s) \bigg]^2 \leq C \sum_{j\neq i}\int_0^t \mathbb{E}\big|\delta_i m_j^{[i]}(s)\big|^2 \,\frac{ds}N \leq  \frac{C}{ N}.\end{equation}

Hence, it remains to control the size of
		\[ \mathbb{E} \bigg[\int_0^t \sum_{j\neq i } \big(m_j- m_j^{(i)}\big) (s)\,dg_{ij}(s)\bigg]^2 = \int_0^t \sum_{j\neq i}\mathbb{E} \big(m_j- m_j^{(i)}\big)^2 (s)\frac{ds}{N}.\]
To this end, we use that $m_j = \langle \sigma_j \rangle = \big \langle m_j^{[i]}\big\rangle$ so that
		\[\mathbb{E} \bigg(\int_0^t \sum_{j\neq i } \big(m_j- m_j^{(i)}\big) (s)\,dg_{ij}(s)\bigg)^2\leq \frac{t}N \sum_{j\neq i}\sup_{s\in [0;t]}\sup_{\sigma_i=\pm1}\mathbb{E} \big(m_j^{[i]}- m_j^{(i)}\big)^2 (s) . \]
Applying once more It\^o's Lemma yields
		\[\big( m_j^{[i]} - m_j^{(i)}\big)(s) =  \sigma_i \int_0^s \sum_{k\neq i } m_{jk}^{[i]} (u) \,dg_{ik}(u) - \int_0^s \sum_{k\neq i } m_k^{[\{i\}] } m_{jk}^{[\{i\}]}(u) \frac{du}{N}, \]
so that the estimate \eqref{eq:L2bnddimjA} implies
		\[ \sup_{s\in [0;t]}\sup_{\sigma_i=\pm1}\mathbb{E} \big(m_j^{[i]}- m_j^{(i)}\big)^2 (s)\leq \frac{C}N.\]
Thus, we find that
		\[\mathbb{E} \bigg(\int_0^t \sum_{j\neq i } \big(m_j- m_j^{(i)}\big) (s)\,dg_{ij}(s)\bigg)^2\leq  \frac{C}N\]
and together with the bounds \eqref{eq:L2TAP1driftbnd}, \eqref{eq:L2TAP1martingale1bnd}, this proves the claim.
\end{proof}

In order to prove the analogue of the hierarchical TAP equations for the two point functions, we also need the bounds from Lemma \ref{lm:L2bndmijk} and the remark following its proof.
\begin{lemma}\label{lm:hTAP2} Let  $0\leq  t < \log 2$ and assume $\epsilon>0$ to be sufficiently small. Then, for some $C = C_{t,\epsilon}>0$, independent of $N$, we have that
		\[ \mathbb{E} \bigg[m_{ij} -  (1-m_i^2) \sum_{k\neq i} g_{ik}m_{kj}^{(i)} \bigg]^2\leq \frac{ C}{N^{1+\epsilon/2}} .\]
\end{lemma}
\begin{proof}
We consider the $ (g_{ik})_{1\leq k\leq N}$ dynamically and use It\^o's Lemma to compute
		\begin{equation}\label{eq:TAP2bnd0}\begin{split} m_j^{[ i ] }   (t) =&\, m_j^{(i ) }   +    \sum_{\substack{ k  \neq i }}\int_0^t  \sigma_i  m_{kj}^{[i]}   (s)\,dg_{ik}(s)   -   \sum_{\substack{ k \neq i} }\int_0^t   m_{k}^{[i]} m_{kj}^{[i]} (s) \frac{ds}N \\
		=&\, m_j^{( i ) }   +    \sum_{\substack{ k  \neq i }}   \sigma_i   g_{ik}m_{kj}^{(i  )}   -   \sum_{\substack{ k \neq i} }    m_{k}^{(i)} m_{kj}^{(i  )} \frac{t}N \\
		&+     \sum_{\substack{ k  \neq i }}\int_0^t  \sigma_i  \big( m_{kj}^{[i]} -m_{kj}^{(i  )}\big)   (s)\,dg_{ik}(s)  \\
		& -   \sum_{\substack{ k \neq i} }\int_0^t  \big( m_{k}^{[i]} m_{kj}^{[i]} (s) -m_{k}^{(i)} m_{kj}^{(i  )} \big)\frac{ds}N.
		\end{split}\end{equation}
If we average the last equation over the spin variable $ \sigma_i\in\{-1,1\}$ and multiply it afterwards by $(1-m_i^2)$, we find with the identity \eqref{eq:condid} that
		\begin{equation}\label{eq:TAP2bnd1}\begin{split}
		  &\Big\| m_{ij} -  (1-m_i^2) \sum_{k\neq i} g_{ik}m_{kj}^{(i)} \Big\|_2\\
		  & \leq   \sup_{\sigma_i =\pm1}  \Big\|   \sum_{\substack{ k  \neq i }}\int_0^t   \big( m_{kj}^{[i]} -m_{kj}^{(i  )}\big)   (s)\,dg_{ik}(s) \Big\|_2\\
		  &\quad  +  \sup_{\sigma_i =\pm1}  \Big\|  \sum_{\substack{ k \neq i} }\int_0^t  \big( m_{k}^{[i]} m_{kj}^{[i]} (s) -m_{k}^{(i)} m_{kj}^{(i  )} \big)\frac{ds}N \Big\|_2\\
		 & \leq \frac C {N^{1/2}}\sup_{s\in [0;t]} \sup_{\substack{  \sigma_i =\pm1}}  \big\|     m_{j}^{[i  ]} (s) - m_{j}^{(i  )} \big\|_2+  \sup_{s\in [0;t]} \sup_{\substack{k\neq i,j\\ \sigma_i =\pm1}}  \big\|       m_{kj}^{[i]}(s) -m_{kj}^{(i  )}   \big\|_2\\
		  & \quad+  \sup_{s\in [0;t]} \sup_{\substack{k\neq i,j\\ \sigma_i =\pm1}} \big\|    \big(m_{k}^{[i]} m_{kj}^{[i]}\big) (s) -m_{k}^{(i)} m_{kj}^{(i  )}    \Big\|_2.
		\end{split}\end{equation}
Hence, let us bound the norms on the right hand side to conclude the claim.

First of all, a straight forward application of Lemma \ref{lm:L2bndmij} and Eq. \eqref{eq:TAP2bnd0} implies that
		\[\frac C {N^{1/2}}\sup_{s\in [0;t]} \sup_{\substack{  \sigma_i =\pm1}}  \big\|     m_{j}^{[i  ]} (s) - m_{j}^{(i  )} \big\|_2\leq \frac{C}{N}.\]
For the two other error terms, we use again It\^o's Lemma which shows that
		\begin{equation}\label{eq:ito2pt}
		\begin{split}
		m_{jk}^{[i  ]}(s) -m_{jk}^{(i  )} =&\;  \sum_{l\neq i } \int_0^s\sigma_i m_{jkl} ^{[i]}(u) \,dg_{il}(u) \\
		& - \sum_{l\neq i }\int_0^s \big( m_{l} ^{[i]}m_{jkl} ^{[i]} +m_{jl} ^{[i]}m_{kl} ^{[i]}\big)(u)\frac{du}N
		\end{split}
		\end{equation}
and, by the product rule, that
		\begin{equation}\label{eq:ito2pt2}\begin{split}
		\big(m_{k}^{[i]}  m_{jk}^{[i  ]}\big) (s) -m_{k}^{(i)} m_{jk}^{(i  )} =&\;  \sum_{l\neq i } \int_0^s\sigma_i m_{jk} ^{[i]}(u)m_{kl} ^{[i]}(u) \,dg_{il}(u) \\
		& - \sum_{l\neq i }\int_0^s \big( m_{l}^{[i]}m_{jk} ^{[i]}m_{kl}^{[i]} \big)(u)\frac{du}N\\
		& +  \sum_{l\neq i } \int_0^s\sigma_i  m_{k}^{[i]}(u) m_{jkl} ^{[i]}(u) \,dg_{il}(u) \\
		& - \sum_{l\neq i }\int_0^s m_{k}^{[i]}(u)\big( m_{l} ^{[i]}m_{jkl} ^{[i]} +m_{jl} ^{[i]}m_{kl} ^{[i]}\big)(u)\frac{du}N\\
		& + \sum_{l\neq i } \int_0^s m_{kl}^{[i  ]}   (u) m_{jkl} ^{[i]}(u) \,\frac{du}N.
		\end{split}\end{equation}
Using Lemmas \ref{lm:L2bndmij}, \ref{lm:L2bndmijk} and the remarks following their proofs, it is simple to check that
		\[\begin{split}
		 & \sup_{s\in [0;t]} \sup_{\substack{k\neq i,j\\ \sigma_i =\pm1}}  \big\|       m_{kj}^{[i]}(s) -m_{kj}^{(i  )}   \big\|_2\\
		  & \quad+  \sup_{s\in [0;t]} \sup_{\substack{k\neq i,j\\ \sigma_i =\pm1}} \big\|   \big(m_{k}^{[i]} m_{kj}^{[i]} (s) -m_{k}^{(i)} m_{kj}^{(i  )} \big)(s)  \Big\|_2\leq \frac{C_\epsilon}{N^{1/2+\epsilon/4}}.
		\end{split}\]
Plugging these estimates into \eqref{eq:TAP2bnd1}, we conclude the lemma.
\end{proof}

We conclude this section with the proof of Theorem \ref{thm:main}.
\begin{proof}[Proof of Theorem \ref{thm:main}. ]
Lemma \ref{lm:hTAP1} establishes the hierarchical TAP equations \eqref{eq:hTAP1}, so it only remains to prove the bound \eqref{eq:hTAP2}. This is a simple consequence of Lemmas \ref{lm:L2bndmij}, \ref{lm:hTAP1} and \ref{lm:hTAP2}. Indeed, using Cauchy-Schwarz we find that
		\[\begin{split}  &\mathbb{E} \bigg[ m_{ij} -  \bigg(1- \tanh^2\Big(h + \sum_{k\neq i} g_{ik}m_k^{(i)}\Big)\bigg) \sum_{l\neq i} g_{il}m_{lj}^{(i)}  \bigg]^2\\
		&\leq C\, \mathbb{E} \bigg[  \bigg(m_i^2- \tanh^2\Big(h + \sum_{k\neq i} g_{ik}m_k^{(i)}\Big)\bigg) \sum_{l\neq i} g_{il}m_{lj}^{(i)}  \bigg]^2 + \frac{C}{N^{1+\epsilon/2}}\\
		&\leq C\,    \Big\|m_i- \tanh\Big(h + \sum_{k\neq i} g_{ik}m_k^{(i)}\Big)\Big\|_2 \Big\|\sum_{k\neq i} g_{ik}m_{kj}^{(i)} \Big\|_4^2 + \frac{C}{N^{1+\epsilon/2}}\leq \frac{C}{N^{1+\epsilon/4}}.
		 \end{split} \]
		
\end{proof}


\section{Proof of the TAP Equations}\label{sec:TAP}

In this section, we prove the bounds \eqref{eq:TAP1} and \eqref{eq:TAP2} from Corollary \ref{cor:main}.

\begin{proof}[Proof of Corollary \ref{cor:main}]

We begin with the proof of \eqref{eq:TAP1}. We have to compute the leading order contribution to
		\[ \sum_{k\neq i} g_{ik} \big(m_k- m_k^{(i)}\big) =:  \sum_{k\neq i} g_{ik} W_k.\]
To compute the leading order, we view $ W_k=W_k(g_{ik}) $ as a function of the coupling $g_{ik}$ and we do a second order Taylor expansion. This implies that
		\[\begin{split} W_k (g_{ik}) 
		&= W_k(g_{ik}=0) +  g_{ik}m_i(1-m_k^2)(g_{ik}=0) -g_{ik}(m_km_{ik})(g_{ik}=0) \\
		& \hspace{0.5cm} + g_{ik}^2\int_0^1ds_1  \int_0^{s_1} ds_2\; \big(\partial_{ik}^2m_k \big)(s_2g_{ik}).
		\end{split}\]
Setting $X_k:= W_k- m_i(1-m_k^2) g_{ik}$, we thus obtain that
		\[ \begin{split}
		X_k  &= W_k(g_{ik}=0) -g_{ik}(m_km_{ik})(g_{ik}=0) - g_{ik}^2\int_0^1 ds\, \big( \partial_{ik} (m_i(1-m_k^2))\big)(sg_{ik})\\
		 &\hspace{0.5cm}+ g_{ik}^2\int_0^1ds_1  \int_0^{s_1} ds_2\; \big(\partial_{ik}^2m_k \big)(s_2g_{ik}).
		 \end{split}\]
Next, let us prove that
 		\begin{equation}\label{eq:TAP1pf1}
		\mathbb{E} \bigg(\sum_{k\neq i} g_{ik} X_k\bigg)^2\leq \frac CN.
		\end{equation}
Since $| \partial_{ik} (m_i(1-m_k^2))(sg_{ik})|\leq C, |\partial_{ik}^2m_k (sg_{ik})|\leq C $ (uniformly in $s\in [0;1]$), we have
		\[\begin{split}
		\mathbb{E} \bigg(\sum_{k\neq i}   g_{ik}^3\int_0^1 ds\, \big( \partial_{ik} (m_i(1-m_k^2))\big)(s g_{ik}) \bigg)^2 &\leq C\, \mathbb{E}  \sum_{k,l\neq i} |g_{ik} |^3|g_{il} |^3 \leq \frac{C}N
		\end{split}\]
as well as
		\[\begin{split}
		\mathbb{E} \bigg(\sum_{k\neq i}   g_{ik}^3\int_0^1ds_1  \int_0^{s_1} ds_2\; \big(\partial_{ik}^2m_k \big)(s_2g_{ik})\bigg)^2 &\leq C\, \mathbb{E}  \sum_{k,l\neq i} |g_{ik} |^3|g_{il} |^3 \leq \frac{C}N.
		\end{split}\]
With the identities
		\begin{equation}\label{eq:TAP1pf2}
		\begin{split}
		\partial_{il} m_{k} &= m_i m_{kl} + m_l m_{ik} +m_{ilk}, \\
		\partial_{il} m_{ik} & = (1-m_i^2)m_{kl} -m_{il}m_{ik}- 2m_im_l m_{ik} - m_im_{ilk},
		\end{split}
		\end{equation}
we then obtain by Gaussian integration by parts
		\[\begin{split}
		&\mathbb{E} \bigg(\sum_{k\neq i}   g_{ik} W_k(g_{ik}=0)\bigg)^2\\
		&\hspace{1cm}= \frac{t}N\mathbb{E} \sum_{k\neq i}W_k^2(g_{ik}=0) + \frac{t^2}{N^2}\mathbb{E}\sum_{k,l\neq i}(\partial_{il} m_k(g_{ik}=0))(\partial_{ik} m_l(g_{il}=0)).
		\end{split}\]
Here, we used that $ \partial_{ik} \big( W_k(g_{ik}=0)\big) =0$ and that $\partial_{il} m_k^{(i)}=0$ (for all $k,l\in\{1,\dots,N\}$). Now, Eq. \eqref{eq:TAP1pf2} and Lemma \ref{lm:L2bndmij} together with the remarks thereafter show that
		\[\begin{split}
		&\frac{t^2}{N^2}\bigg| \mathbb{E}\sum_{k,l\neq i}(\partial_{il} m_k(g_{ik}=0))(\partial_{ik} m_l(g_{il}=0))\bigg|\\
		&\leq \frac{t^2}{N^2}\mathbb{E}\sum_{k,l\neq i} \big( m_i m_{kl} + m_l m_{ik} +m_{ilk}\big)^2(g_{ik}=0)\leq  \frac{C}N.
		\end{split}\]	
Notice that applying Lemma \ref{lm:L2bndmij} is enough to obtain the previous bound, because we can bound the $L^2(\mathbb{P})$ norms of the three point functions $m_{ikl}$ by the $L^2(\mathbb{P})$ norms of suitable two point functions, through the identity \eqref{eq:condid2}.

To estimate $ tN^{-1}\mathbb{E} \sum_{k\neq i}W_k^2(g_{ik}=0)$, on the other hand, we recall Eq. \eqref{eq:TAP2bnd0} so that
		\[\begin{split} W_k(g_{ik}=0) &= \Big(\big\langle m_k^{[i]}\big\rangle-m_k^{(i)}\Big)(g_{ik}=0)\\
		 &= \sum_{j\neq i, k} \bigg(\int_0^t \big\langle \sigma_i m_{jk}^{[i]}\big\rangle(s)\, dg_{ij}(s) -  \int_0^t \big\langle m_{j}^{[i]}m_{jk}^{[i]}\big\rangle(s)\frac{ds}N\bigg)(g_{ik}=0).
		 \end{split} \]
Hence, Lemma \ref{lm:L2bndmij} implies also in this case that
		\[ \frac{t}N\mathbb{E} \sum_{k\neq i}W_k^2(g_{ik}=0) \leq \frac CN\]
and, similarly, for the remaining contribution that
		\[\begin{split}
		\mathbb{E} \bigg(\sum_{k\neq i}   g_{ik}^2(m_km_{ik})(g_{ik}=0) \bigg)^2&\leq \mathbb{E} \bigg( \sum_{k\neq i}   g_{ik}^4\bigg)\bigg( \sum_{k\neq i}  m_{ik}^2(g_{ik}=0)\bigg)\\
		&=  \mathbb{E} \bigg( \sum_{k\neq i}   g_{ik}^4\bigg) \,\mathbb{E}\bigg( \sum_{k\neq i}  m_{ik}^2(g_{ik}=0)\bigg) \leq \frac{C}N.
		\end{split}\]
Collecting the previous bounds, we  conclude that
		\begin{equation}\label{eq:TAP1pf3}\begin{split}
		 \sum_{k\neq i} g_{ik}m_k^{(i)} &= \sum_{k\neq i} g_{ik}m_k - \sum_{k\neq i} g_{ik}^2 (1-m_k^2) m_i - \sum_{k\neq i}g_{ik}X_k\\
		 & = \sum_{k\neq i} g_{ik}m_k - t(1-q_N) m_i  - \sum_{k\neq i} \big(g_{ik}^2-t/N\big) (1-m_k^2) m_i - \sum_{k\neq i}g_{ik}X_k.
		 \end{split}\end{equation}
Here, the error term $\sum_{k\neq i}g_{ik}X_k$ satisfies the estimate \eqref{eq:TAP1pf1} and, arguing once more as above, we also find that
		\[\begin{split}
		&\mathbb{E}\bigg(  \sum_{k\neq i}\big( g_{ik}^2-t/N\big) (1-m_k^2) m_i\bigg)^2\\
		& = \mathbb{E}   \sum_{ k,l\neq i:k\neq l}\big( g_{ik}^2g_{il}^2-2g_{ik}^2t/N+t^2/N^2\big) (1-m_k^2) (1-m_l^2) m_i^2\\
		&\hspace{0.5cm}+\mathbb{E}   \sum_{ k\neq i }\big( g_{ik}^4 -2g_{ik}^2t/N+t^2/N^2\big) (1-m_k^2) ^2  m_i^2 \\
		& \leq  \mathbb{E}  \frac{t}N \sum_{ k,l\neq i:k\neq l}\big( g_{ik} g_{il}^2-2g_{ik}t/N \big) \partial_{ik}\Big[(1-m_k^2) (1-m_l^2) m_i^2\Big] \\
		&\hspace{0.5cm}+\mathbb{E}  \frac{t^2}{N^2} \sum_{k,l\neq i: k\neq l} g_{il} \partial_{il}  \Big[(1-m_k^2) (1-m_l^2) m_i^2\Big] + \frac{C}N\leq \frac{C}N.
		\end{split}\]
Note that the last bound follows from repeated Gaussian integration by parts and the fact that derivatives of $ (1-m_k^2) m_i $ are bounded by some constant $C>0$. By the Lipschitz continuity of $x\mapsto \tanh(x)$, this proves with Eq. \eqref{eq:TAP1pf3} the TAP equations \eqref{eq:TAP1}.

Let us now turn to the proof of the TAP equations \eqref{eq:TAP2} for the two point functions. We use the same ideas as for the proof of Eq. \eqref{eq:TAP1} and focus on the main steps. By Lemma \ref{eq:hTAP2}, we have to determine the leading order contribution to
		\[ \sum_{k\neq i }g_{ik}\big(m_{kj}- m_{kj}^{(i)} \big)=: \sum_{k\neq i }g_{ik} Y_k. \]
We view $Y_k= Y_k(g_{ik})$ as a function of $g_{ik}$ and a second order Taylor expansion yields
		\[\begin{split}
		Y_k(g_{ik}) &=  Y_k(g_{ik}=0) + g_{ik} \big( (1-m_k^2)m_{ij}  -2m_i m_k m_{kj}\big)(g_{ik}=0)\\
		&\hspace{0.5cm} - g_{ik}\big( m_{ik}m_{jk}+ m_k m_{ijk}\big)(g_{ik}=0)  + g_{ik}^2 \int_0^1 ds_1 \int_0^{s_1}ds_2\, \big(\partial_{ik}^2 m_{kj} \big)(s_2g_{ik}).
		\end{split}\]
Hence, defining $Z_k := Y_k - g_{ik}   (1-m_k^2)m_{ij}  + 2g_{ik} m_i m_k m_{kj}$, we find
		\begin{equation}\label{eq:defZk}\begin{split}
		Z_k &= Y_k(g_{ik}=0) - g_{ik}^2\int_0^1ds\, \partial_{ik}\big( (1-m_k^2)m_{ij}  -2m_i m_k m_{kj}\big)(sg_{ik})\\
		&\hspace{0.5cm} - g_{ik}\big( m_{ik}m_{jk}+ m_k m_{ijk}\big)(g_{ik}=0)  + g_{ik}^2 \int_0^1 ds_1 \int_0^{s_1}ds_2\, \big(\partial_{ik}^2 m_{kj} \big)(s_2g_{ik}).
		\end{split}\end{equation}
		
Now, in the first step, we prove that for all $\epsilon>0$ sufficiently small, it holds true that
		\begin{equation}\label{eq:TAP2pf1}
		\mathbb{E} \bigg(\sum_{k\neq i} g_{ik} Z_k\bigg)^2\leq \frac C{N^{1+\epsilon}}.
		\end{equation}
This follows from the decay results of Lemma \ref{lm:L2bndmij}, \ref{lm:L2bndmijk} and the remarks following their proofs. We start with the term
		\[\begin{split}
		&\mathbb{E} \bigg(\sum_{k\neq i} g_{ik}^3 \int_0^1ds\, \partial_{ik}\big( (1-m_k^2)m_{ij}  -2m_i m_k m_{kj}\big)(sg_{ik}) \bigg)^2\\
		&\leq C\sup_{s\in [0;1] }  \mathbb{E}  \sum_{k,l\neq i} |g_{ik}^3||g_{il}^3| |m_{ij}(sg_{ik})|^2  + C\sup_{s\in [0;1] }  \mathbb{E} \sum_{k,l\neq i,j} |g_{ik}^3||g_{il}^3||m_{kj}(sg_{ik})|^2+ \frac{C}{N^2}\\
		& \hspace{0.3cm}+ C\sup_{s\in [0;1] }  \mathbb{E}  \sum_{k,l\neq i,j} |g_{ik}^3||g_{il}^3| |\partial_{ik}m_{ij}(sg_{ik})|^2   + C\sup_{s\in [0;1] }  \mathbb{E} \sum_{k,l\neq i,j} |g_{ik}^3||g_{il}^3||\partial_{ik} m_{kj}(sg_{ik})|  + \frac{C}{N^2}\\
		&\leq \frac{C}{N^{3/2}},
		\end{split}\]
where we recall that we assume $i\neq j$ and where we used the identity
		\begin{equation}\label{eq:TAPpf2} \begin{split}
		\partial_{ik} m_{kj} & = -2 m_j \big( m_i m_{jk} + m_k m_{ij} +m_{ijk}\big)\big(\delta_{j} m_k^{[j]}\big)\\
		&\hspace{3cm} + (1-m_j^2) \delta_j\Big[ \big( 1-(m_k^{[j]})^2\big) m_i^{[j]} - m_k^{[j]}m_{ik}^{[j]} \Big]
		\end{split}\end{equation}
to control the terms involving $ \partial_{ik}m_{ij}$ and $\partial_{ik} m_{kj}$. Observe that Eq. \eqref{eq:TAPpf2} is a simple consequence of the conditional identity \eqref{eq:condid}. Notice also that, here and in the following, we frequently use rough bounds of the form $\mathbb{E}\, m_{ij}^4\leq C\, \mathbb{E}\,m_{ij}^2 $ so that all of the following bounds hold true for times $t<\log 2$.

Analogously to the last bound, we obtain that
		\[\begin{split}
		&\mathbb{E} \bigg(\sum_{k\neq i}g_{ik}^2\big( m_{ik}m_{jk}+ m_k m_{ijk}\big)(g_{ik}=0)\bigg)^2\\
		&\leq C \sum_{k,l\neq i,j}\mathbb{E}  g_{ik}^4  \,\mathbb{E} \big(m_{ik}^2m_{jk}^2\big)(g_{ik}=0)+C   \sum_{k,l\neq i,j} \mathbb{E} g_{ik}^4\, \mathbb{E} \big( m_{ijk}^2\big) (g_{ik}=0) + \frac{C}{N^2}\leq \frac{C}{N^{1+\epsilon}}.
		\end{split}\]

To bound the last contribution on the right hand side of Eq. \eqref{eq:defZk}, we differentiate the identity \eqref{eq:TAPpf2} and a tedious, but straight forward computation shows that
		\[\begin{split}
		\partial_{ik}^2 m_{kj} & = -2  \big( m_i m_{jk} + m_k m_{ij} +m_{ijk}\big)^2\big(\delta_{j} m_k^{[j]}\big)\\
		&\hspace{0.5cm} - 2 m_j \Big[ \partial_{ik}\big( m_i m_{jk} + m_k m_{ij} +m_{ijk}\big)\Big]\big(\delta_{j} m_k^{[j]}\big)\\
		&\hspace{0.5cm} -2 m_j \big( m_i m_{jk} + m_k m_{ij} +m_{ijk}\big) \Big[ \partial_{ik}\big(\delta_{j} m_k^{[j]}\big)\Big]\\
		&\hspace{0.5cm}  -2m_j\big( m_i m_{jk} + m_k m_{ij} +m_{ijk}\big) \delta_j\Big[ \big( 1-(m_k^{[j]})^2\big) m_i^{[j]} - m_k^{[j]}m_{ik}^{[j]}\Big]\\
		&\hspace{0.5cm} + (1-m_j^2)  \delta_j \Big[ -2 m_k^{[j]}\big(1-(m_k^{[j]})^2\big)\big(m_i^{[j]} \big)^2  -2 m_i^{[j]}\big(m_k^{[j]}\big)^2m_{ik}^{[j]}   \Big]\\
		&\hspace{0.5cm} -(1-m_j^2) \delta_j\Big[  \big(1-(m_k^{[j]})^2\big)m_i^{[j]}+ 2m_k^{[j]} (m_{ik}^{[j]})^2 + \big(1-(m_k^{[j]})^2\big)m_i^{[j]}m_{ik}^{[j]}  \Big].
		\end{split} \]
 If we then proceed as above, using the bounds from Lemmas \ref{lm:L2bndmij} and \ref{lm:L2bndmijk} combined with the product rule for the action of $\delta_j$ (in the last formula), we verify that
 		\[\begin{split}
		&\mathbb{E} \bigg( \sum_{k\neq i} g_{ik}^3 \int_0^1 ds_1 \int_0^{s_1}ds_2\, \big(\partial_{ik}^2 m_{kj} \big)(s_2g_{ik})\bigg)^2\\
		&\leq C \sup_{s\in [0;1]}\mathbb{E}   \sum_{k,l\neq i,j} g_{ik}^6   \big(\partial_{ik}^2 m_{kj} \big)^2(sg_{ik}) + \frac{C}{N^2}\leq \frac{C}{N^{3/2}}.
		\end{split}\]				
Finally, it remains to bound the first term on the right hand side in Eq. \eqref{eq:defZk}. We have
		\[\begin{split}
		&\mathbb{E} \bigg( \sum_{k\neq i} g_{ik} Y_k(g_{ik}=0)\bigg)^2\\
		&= \mathbb{E}\frac{t}N\sum_{k\neq i} Y_k^2(g_{ik}=0) + \mathbb{E}\frac{t^2}{N^2}\sum_{k,l\neq i:k\neq l} (\partial_{ik}m_{lj})(g_{il}=0)(\partial_{il}m_{kj})(g_{ik}=0)\\
		&\hspace{0.5cm} +  \mathbb{E}\frac{t^2}{N^2}\sum_{k\neq i} (\partial_{ik}m_{kj})^2(g_{ik}=0) \\
		& \leq \mathbb{E}\frac{t}N\sum_{k\neq i} Y_k^2(g_{ik}=0) + \mathbb{E}\frac{t^2}{N^2}\sum_{k,l\neq i:k\neq l} (\partial_{ik}m_{lj})^2(g_{il}=0) +\frac{C}{N^2},
		\end{split}\]
where we used Eq. \eqref{eq:TAPpf2} to obtain the estimate of the last line. Recalling the identity \eqref{eq:ito2pt}, it is furthermore straight forward to show that
		\[\mathbb{E}\frac{t}N\sum_{k\neq i} Y_k^2(g_{ik}=0) \leq \frac{C}{N^{1+\epsilon}}\]
and the smallness of the last contribution follows from the identity
		\begin{equation}\label{eq:TAPpf3}\begin{split}
		\partial_{ik}m_{lj} &= -2m_l \big(  m_i m_{kl} + m_k m_{il} +m_{ilk} \big) \big(\delta_l m_j^{[l]}\big) \\
		&\hspace{3cm}+(1-m_l^2)\delta_l \big(  m_i^{[l]} m_{kj} ^{[l]}+ m_k^{[l]} m_{ij}^{[l]} +m_{ijk}^{[l]} \big).
		\end{split}\end{equation}
It implies with the product rule for $\delta_l$ and the identity \eqref{eq:condid2} that
		\[\begin{split}
		 \mathbb{E}\frac{t^2}{N^2}\sum_{k,l\neq i:k\neq l} (\partial_{ik}m_{lj})^2(g_{il}=0)\leq \frac{C}{N^{1+\epsilon}}.
		\end{split}\]

Collecting the previous estimates, we summarize that we have shown that
		\[  \sum_{k\neq i} g_{ik}m_{kj}^{(i)}  = \sum_{k\neq i} g_{ik}m_{kj}  + 2\sum_{k\neq i} g_{ik}^2   m_{jk}m_km_i- \sum_{k\neq i} g_{ik}^2   (1-m_k^2)m_{ij}  - \sum_{k\neq i} g_{ik} Z_k,\]
where the error $ \sum_{k\neq i}g_{ik}Z_k$ satisfies the estimate \eqref{eq:TAP2pf1}. To conclude the TAP equations \eqref{eq:TAP2}, it thus only remains to replace $g_{ik}^2$ by its mean in the previous equation and to show that the resulting error is small. To this end, we apply once more the arguments from the previous steps to deduce that
		\[\begin{split}
		\mathbb{E} \bigg( \sum_{k\neq i} \big(g_{ik}^2-t/N\big)   m_{jk}m_km_i \bigg)^2 + \mathbb{E} \bigg( \sum_{k\neq i} \big(g_{ik}^2-t/N\big)  (1-m_k^2)m_{ij}\bigg)^2\leq \frac{C}{N^{1+\epsilon}}.
		\end{split}\]
We omit the details and conclude the proof of Corollary \eqref{cor:main}.
\end{proof}




\section{Overlap Concentration and Computation of $ \mathbb{E} \,m_{ij}^2$}\label{sec:conc}
In this section, we outline the proofs of Propositions \ref{prop:concentration} and \ref{prop:mij2norm}. Let us start with the proof of the concentration of the overlap, Eq. \eqref{eq:concentration}. This is a consequence of the TAP equations \eqref{eq:hTAP1} for the magnetizations $m_i$ and follows from a contraction argument.

\begin{proof}[Proof of Proposition \ref{prop:concentration}.]
Let $Z \sim \mathcal{N}(0,1)$ denote a standard Gaussian random variable, independent of the disorder $(g_{ij})_{1\leq i<j\leq N}$. We define $f: [0;\infty) \to [0;\infty)$ through	
		\[ f(x) =  \mathbb{E}_Z\tanh^2( h + \sqrt{tx}Z), \]
where $\mathbb{E}_Z$ denotes the expectation with respect to the randomness of $Z$. By Gaussian integration by parts, we find that
		\[ f'(x) = t\,\mathbb{E}_Z \frac{1- 2\sinh^2(h + \sqrt{tx}Z) }{\cosh^4( h + \sqrt{tx}Z)},\]
and therefore that
		\[ \sup_{x\in [0;\infty) }| f'(x)| \leq t\sup_{y\in [0;\infty) } \bigg| \frac{1- 2\sinh^2(y)}{\cosh^4( y)}\bigg|\leq t. \]
This follows from $ \cosh^2(y)\geq 1$ and
		\[ 2\tanh ^2(y) \leq 2 \leq \frac1{\cosh^2(y)} + \cosh^2(y). \]
In particular $f$ is Lipschitz continuous with Lipschitz constant bounded by $t<\log 2<1$.
		
Next, let us also recall that $q_N = N^{-1} \sum_{k=1}^N m_k^2$. By Eq. \eqref{eq:hTAP1}, we have that
		\[   m_i =  \tanh\Big( h + \sum_{k\neq i}g_{ik} m_k^{(i)}\Big) + \Phi_i, \]
where $ \mathbb{E}\,\Phi_i^2\leq C/N $. This implies, by symmetry, that
		\[\begin{split}
		 &\Big| \mathbb{E} \,q_N - \mathbb{E} \,\tanh^2\Big( h + \sum_{k\neq 1}g_{1k} m_k^{(1)}\Big) \Big| \leq \frac{C}{N^{1/2}}, \\
		 &\Big| \mathbb{E} \,q_N^2 - \mathbb{E} \,\tanh^2\Big( h + \sum_{k\neq 1}g_{1k} m_k^{(1)}\Big)\tanh^2\Big( h + \sum_{k\neq 2}g_{2k} m_k^{(2)}\Big) \Big| \leq \frac{C}{N^{1/2}}.
		 \end{split}\]
		
Now, proceeding as in Section \ref{sec:hTAP}, it is straight forward to verify that
		\[\begin{split}
		 &\mathbb{E} \,\bigg[  \tanh^2\Big( h + \sum_{k\neq 1}g_{1k} m_k^{(1)}\Big) - \tanh^2\Big( h + \sum_{k\neq 1,2}g_{1k} m_k^{( 1,2)}\Big)\bigg]^2\\
		 &\leq  \mathbb{E}\,\bigg[  g_{12}m_2^{(1)}+  \sum_{k\neq 1,2}g_{1k} \big( m_k^{(1)} - m_k^{( 1,2)}\big)\Big)\bigg]^2  \leq  \mathbb{E}\, \frac{C}N\sum_{k\neq 1,2} \big( m_k^{(1)} - m_k^{( 1,2)}\big)^2 + \frac{C}N\leq \frac{C}N,
		 \end{split}\]
where, by slight abuse of notation, we abbreviate from now on $ m_k^{( 1,2)}:= m_k^{(\{1,2\})}$. Observe that the last bound follows from It\^o's lemma applied to $(g_{2k})_{1\leq k\leq N}$. We have similarly
		\[\begin{split}
		 &\mathbb{E} \,\bigg[  \tanh^2\Big( h + \sum_{k\neq 2}g_{2k} m_k^{(2)}\Big) - \tanh^2\Big( h + \sum_{k\neq 1,2}g_{2k} m_k^{( 1,2)}\Big)\bigg]^2 \leq \frac{C}N.
		 \end{split}\]
Since the last two estimates can be proved with the same arguments as in Sections \ref{sec:corrdec} and \ref{sec:hTAP}, we skip the details. What they imply is that
		\[\begin{split}
		 &\Big| \mathbb{E} \,q_N - \mathbb{E} \,\tanh^2\Big( h + \sum_{k\neq 1}g_{1k} m_k^{(1)}\Big) \Big| \leq \frac{C}{N^{1/2}}, \\
		 &\Big| \mathbb{E} \,q_N^2 - \mathbb{E} \,\tanh^2\Big( h + \sum_{k\neq 1,2}g_{1k} m_k^{(1,2)}\Big)\tanh^2\Big( h + \sum_{k\neq 1,2}g_{2k} m_k^{(1,2)}\Big) \Big| \leq \frac{C}{N^{1/2}}.
		 \end{split}\]
Now, setting $ q_N^{(1)} = N^{-1}\sum_{k\neq 1}\big(m_k^{(1)}\big)^2$, we have (as observed in \cite[Lemma 1.7.6]{Tal1}) that
		\[ Z_1 :=   \big(tq_N^{(1)}\big)^{-1/2} \sum_{k\neq 1} g_{1k}m_k^{(1)} \sim \mathcal{N}(0,1) \]
is independent of $g_{kl}$ for all $k,l\neq 1$ (and hence unconditionally Gaussian). Therefore
		\[\mathbb{E} \,\tanh^2\Big( h + \sum_{k\neq 1}g_{1k} m_k^{(1)}\Big) = \mathbb{E}\, f \big(q_N^{(1)}\big).  \]
Similarly, defining $q_N^{(1,2)} = N^{-1}\sum_{k\neq 1,2}\big(m_k^{(1,2)}\big)^2$ as well as the Gaussians
		\[\begin{split}
		&Z_{12} :=   \big(tq_N^{(1,2)}\big)^{-1/2} \sum_{k\neq 1,2} g_{1k}m_k^{(1,2)}  \sim \mathcal{N}(0,1), \\
		&Z_{22} :=   \big(tq_N^{(1,2)}\big)^{-1/2} \sum_{k\neq 1,2} g_{2k}m_k^{(1,2)}  \sim \mathcal{N}(0,1),
		\end{split}\]
we easily see that
		\[ \mathbb{E}_{g_{1\bullet}g_{2\bullet}} Z_{12}^2 = 1, \hspace{0.5cm} \mathbb{E}_{g_{1\bullet}g_{2\bullet}}\,Z_{22}^2 = 1, \hspace{0.5cm} \mathbb{E}_{g_{1\bullet}g_{2\bullet}} \, Z_{12}Z_{22} = 0. \]
Here, $ \mathbb{E}_{g_{1\bullet}g_{2\bullet}}$ denotes the expectation conditionally on $g_{kl}$ for all $k,l\neq 1,2$. Thus, $ Z_{12}$ and $ Z_{22}$ are, conditionally on $g_{kl}$ for all $k,l\neq 1,2$, two i.i.d. standard Gaussians. Since their conditional statistics is deterministic, $ (Z_{12}, Z_{22})\sim \mathcal{N}(0, \textbf{1}_{\mathbb{R}^2 })$ is unconditionally jointly Gaussian, and independent of the remaining disorder $g_{kl}$ for all $k,l\neq 1,2$.

As in the previous step, we therefore find that
		\[ \begin{split}
		&\mathbb{E} \,\tanh^2\Big( h + \sum_{k\neq 1,2}g_{1k} m_k^{(1,2)}\Big)\tanh^2\Big( h + \sum_{k\neq 1,2} g_{2k} m_k^{(1,2)}\Big)\\
		& = \mathbb{E} \, \mathbb{E}_{g_{1\bullet}g_{2\bullet}} \,\tanh^2 \Big( h + \sqrt{t q_N^{(1,2)}}Z_{12}\Big)\tanh^2 \Big( h + \sqrt{t q_N^{(1,2)}}Z_{22}\Big)
		= \mathbb{E} \,f^2 \big(q_N^{(1,2)}\big)
		\end{split} \]
Finally, let us point out that the Lipschitz continuity of $f$ implies that
		\[\begin{split} \big| \mathbb{E} \,f \big( q_N^{(1)}\big) -\mathbb{E} \,f \big( q_N \big) \big| &\leq \big\| m_2-m_2^{(1)}  \big\|_2 + \frac{C}{N^{1/2}}\leq \frac{C}{N^{1/2}}, \\
		\big| \mathbb{E} \,f^2 \big( q_N^{(1,2)}\big) -\mathbb{E} \,f^2 \big( q_N \big) \big|&\leq 2 \big\| m_3-m_3^{(1)}  \big\|_2+ 2\big\| m_3^{(1)}-m_3^{(1,2)}  \big\|_2 + \frac{C}{N^{1/2}}\leq \frac{C}{N^{1/2}}.
		\end{split}\]

Collecting the above observations, we obtain that
		\[ \begin{split}
		\mathbb{E} \big| q_N - \mathbb{E}\,q_N\big|^2 \leq &\;\mathbb{E} \big( q_N - f(\mathbb{E}\, q_N)\big)^2 = \mathbb{E}\, q_N^2 -2 f(\mathbb{E}\,q_N) \mathbb{E} \,q_N + f^2(\mathbb{E}\,q_N)\\
		\leq &\; \mathbb{E} \, f^2( q_N) -2 f( \mathbb{E} \,q_N) \mathbb{E} \, f(q_N) + f^2(\mathbb{E}\,q_N) +\frac{C}{N^{1/2}} \\
		=&\; \mathbb{E} \big|f(q_N) - f(\mathbb{E} \,q_N)\big|^2 + \frac{C}{N^{1/2}}\leq \sup_{x\in [0;\infty)} |f'(x)|^2\, \mathbb{E} \big| q_N - \mathbb{E}\,q_N\big|^2 + \frac{C}{N^{1/2}}.
		\end{split}  \]
Since $\sup_{x\in [0;\infty)} |f'(x)|^2\leq t^2<1$, this proves that $q_N$ concentrates, i.e.
		\[\mathbb{E} \big|q_N - \mathbb{E}\,q_N\big|^2 \leq \frac{C}{N^{1/2}}.\]
Using again the Lipschitz continuity of $f$, it also shows that
 		\[   \big| \mathbb{E}\,q_N  - f\big(\mathbb{E}\,q_N \big)\big| \leq \big| \mathbb{E}\,f(q_N)  - f\big(\mathbb{E}\,q_N \big)\big|+ \frac{C}{N^{1/2}}\leq \frac{C}{N^{1/4}}.  \]
If $q \in [0;1]$ denotes the unique fixed point $ q = \mathbb{E}_Z \tanh^2(h+ \sqrt{tq}Z) = f(q)$ (for the uniqueness, see for instance \cite[Prop. 1.3.8]{Tal1} and recall that $t<1$), we conclude that
		\[ \big| q- \mathbb{E}\,q_N\big| \leq \big| f(q) - f(\,\mathbb{E}\,q_N)\big| + \frac{C}{N^{1/4}} \leq t\big| q- \mathbb{E}\,q_N\big| +\frac{C}{N^{1/4}}, \]
so that $\big| q- \mathbb{E}\,q_N\big|\leq C/ N^{1/4} $. This implies in particular \eqref{eq:concentration} and finishes the proof.
\end{proof}

Having proved the concentration of the overlap, let us now make the heuristics \eqref{eq:heuristic} rigorous in order to prove Proposition \ref{prop:mij2norm}. Before we start, we record that
		\begin{equation} \label{eq:concentrationLp} \mathbb{E} \big|q - q_N\big|^p \leq  \frac{C}{N^{1/2}} \end{equation}
for any $p\geq 2$, which follows by interpolation from the concentration bound \eqref{eq:concentration} and the boundedness of $q_N = N^{-1}\sum_{k=1}^Nm_k^2  \leq 1 $.

\begin{proof}[Proof of Proposition \ref{prop:mij2norm}.]
By the TAP equations \eqref{eq:hTAP2} and Gaussian integration by parts, we find that 	
		\[\begin{split}
		\mathbb{E}\, m_{ij}^2 =  &\, \mathbb{E} \,t  \,  \sech^{4} \Big(h + \sum_{k\neq i} g_{ik}m_k^{(i)} \Big)\frac1N \sum_{l\neq i} \big(m_{lj}^{(i)}\big)^2 \\
		&+ \mathbb{E}\frac{4t^2}{N^2}\bigg(\sum_{l\neq i}m_{l}^{(i)} m_{lj}^{(i)}  \bigg)^2 \, \frac{\big(4\sinh^2\big(h + \sum_{k\neq i} g_{ik}m_k^{(i)} \big)-1\big)}{\cosh^6\big(h + \sum_{k\neq i} g_{ik}m_k^{(i)} \big)}  + \Theta_1,
		\end{split}\]
where the error $ \Theta_1$ satisfies $| \Theta_1 |\leq C/N^{1+\epsilon}$, for $\epsilon >0$ sufficiently small. To estimate the first term in the second line, we use that $\sup_{x\in \mathbb{R}}\big| \frac{4\sinh^2(x)-1}{\cosh^6(x)} \big|  \leq C$ so that
		\begin{equation}\label{eq:app1}\begin{split}
		&\bigg|\mathbb{E}\frac{4t^2}{N^2}\bigg(\sum_{l\neq i}m_{l}^{(i)} m_{lj}^{(i)}  \bigg)^2 \, \frac{\big(4\sinh^2\big(h + \sum_{k\neq i} g_{ik}m_k^{(i)} \big)-1\big)}{\cosh^6\big(h + \sum_{k\neq i} g_{ik}m_k^{(i)} \big)} \bigg|\\
		& \leq \mathbb{E}\frac{4t^2}{N^2}  \bigg(\sum_{l\neq i} m_{l}^{(i)} m_{lj}^{(i)}  \bigg)^2 \bigg|\frac{\big(4\sinh^2\big(h + \sum_{k\neq i} g_{ik}m_k^{(i)} \big)-1\big)}{\cosh^6\big(h + \sum_{k\neq i} g_{ik}m_k^{(i)} \big)}\bigg|  \\
		&\leq C\,  \mathbb{E} \bigg(\frac1N\sum_{l\neq i,j} m_{l}^{(i)} m_{lj}^{(i)}  \bigg)^2 +\frac{C}{N^{3/2}},
		\end{split} \end{equation}
where the last bound follows from Lemma \ref{lm:L2bndmij}.

To continue, we control the first term on the right hand side of the last equation through another contraction argument. This term is an expectation over mixed correlation functions and we are going to show that this term is of lower order $o(N^{-1})$, as claimed in \eqref{eq:heuristic}. To make this rigorous, it is first of all useful to observe that
		\begin{equation}\label{eq:app2} \mathbb{E} \bigg(\frac1N\sum_{l\neq i,j} \Big[m_{l}  m_{lj}  -  m_{l}^{(i)} m_{lj}^{(i)}\Big]  \bigg)^2\leq \frac{C}{N^{1+\epsilon}}.\end{equation}
This can be proved using the results of Lemma \ref{lm:L2bndmij} and \ref{lm:L2bndmijk}, proceeding as in Section \ref{sec:hTAP} (recall in particular Eq. \eqref{eq:ito2pt2}); we omit the details. By Lemma \ref{lm:hTAP2}, we then see that
		\[ \mathbb{E} \bigg(\frac1N\sum_{l\neq i,j}  m_{l}  m_{lj}\bigg)^2 = \mathbb{E} \bigg(\frac1N\sum_{l\neq i,j}  m_{l}  (1-m_j^2)\sum_{k\neq j}g_{jk}m_{kl}^{(j)}  \bigg)^2 + \Theta_2 \]
with an error $\Theta_2$ such that $ |\Theta_2|\leq C/N^{1+\epsilon}$. Since we can pull the non--negative factor $ (1-m_j^2)\leq 1$ out of the summations, we find that
		\[ \begin{split}
		&\mathbb{E} \bigg(\frac1N\sum_{l\neq i,j}  m_{l}  (1-m_j^2)\sum_{k\neq j}g_{jk}m_{kl}^{(j)}  \bigg)^2\leq \mathbb{E} \bigg(\frac1N\sum_{l\neq i,j}  m_{l}  \sum_{k\neq j}g_{jk}m_{kl}^{(j)}  \bigg)^2\\
		& =\mathbb{E}\frac{t}{N^3}    \sum_{\substack{ l_1,l_2\neq i,j; \\ k\neq j } }  m_{l_1}m_{l_2}   m_{kl_1}^{(j)}m_{kl_2}^{(j)} + \mathbb{E}\frac{t^2}{N^4}   \sum_{\substack{ l_1,l_2\neq i,j; \\ k_1,k_2\neq j } } m_{k_1l_1}^{(j)}m_{k_2l_2}^{(j)}  \partial_{jk_1} \partial_{jk_2}\big(m_{l_1}m_{l_2} \big) \\
		& = \mathbb{E}\frac{t}{N^2}   \bigg(\sum_{\substack{ l\neq i,j,r  } }  m_{l}    m_{lr}^{(j)}\bigg)^2 +\mathbb{E} \frac{t^2}{N^4}\!\!\!  \sum_{\substack{ l_1,l_2\neq i,j; \\ k_1,k_2\neq j,l_1,l_2; \\k_1\neq k_2, l_1\neq l_2} } \!\!\! m_{k_1l_1}^{(j)}m_{k_2l_2}^{(j)} \partial_{jk_1} \partial_{jk_2}\big(m_{l_1}m_{l_2} \big)  +\Theta_3
		\end{split}\]
for an error $\Theta_3$ such that $ |\Theta_3| \leq C/N^{1+\epsilon}$ and some fixed $r\neq j $, by symmetry. But then, on the one hand, we can use Eq. \eqref{eq:TAP1pf2}, the identity \eqref{eq:condid2} and Eq. \eqref{eq:TAPpf3} to deduce that
		\[\begin{split}
		&\mathbb{E} \frac{t^2}{N^4}\!\!\!  \sum_{\substack{ l_1,l_2\neq i,j; \\ k_1,k_2\neq j,l_1,l_2; \\k_1\neq k_2, l_1\neq l_2 } } \!\!\! m_{k_1l_1}^{(j)}m_{k_2l_2}^{(j)} \partial_{jk_1} \partial_{jk_2}\big(m_{l_1}m_{l_2} \big)\\
		& =   \mathbb{E} \frac{2t^2}{N^4}\!\!\!  \sum_{\substack{ l_1,l_2\neq i,j; \\ k_1,k_2\neq j,l_1,l_2; \\k_1\neq k_2, l_1\neq l_2 } } \!\!\! m_{k_1l_1}^{(j)}m_{k_2l_2}^{(j)} \partial_{jk_1} \Big[ m_{l_1} \big( m_j m_{k_2l_2}+m_{k_2}m_{jl_2} + m_{jk_2l_2} \big) \Big]\leq \frac{C}{N^{1+\epsilon}}.
		\end{split}\]
On the other hand, we find with similar arguments as before that
		\[ \mathbb{E}    \bigg(\frac1N\sum_{\substack{ l\neq i,j,r  } }  m_{l}   \big(  m_{lr}-m_{lr}^{(j)}\big)\bigg)^2 \leq \frac{C}{N^{1+\epsilon}}.\]
Therefore, if we combine the previous bounds, we have shown that
		\[ \begin{split}
		\mathbb{E} \bigg(\frac1N\sum_{l\neq i,j}  m_{l}  m_{lj}\bigg)^2 &\leq t\,\mathbb{E}    \bigg(\frac1N\sum_{\substack{ l\neq i,j,r  } }  m_{l}    m_{lr}\bigg)^2+\frac{C}{N^{1+\epsilon}}\leq t\,\mathbb{E}    \bigg(\frac1N\sum_{\substack{ l\neq i,j  } }  m_{l}    m_{lj}\bigg)^2+\frac{C}{N^{1+\epsilon}},
		\end{split}\]
and we conclude under the assumption $t<\log 2<1$ that
		\[\mathbb{E} \bigg(\frac1N\sum_{l\neq i,j}  m_{l}  m_{lj}\bigg)^2 \leq \frac{C(1-t)^{-1}}{ N^{1+\epsilon}}\leq \frac{C}{N^{1+\epsilon}}. \]
By Eq. \eqref{eq:app2}, this also implies that
		\[\mathbb{E} \bigg(\frac1N\sum_{l\neq i,j}  m_{l}^{(i)}  m_{lj}^{(i)}\bigg)^2 \leq \frac{C}{N^{1+\epsilon}} \]
and plugging this into Eq. \eqref{eq:app1}, it follows that
		\[\bigg|\mathbb{E}\frac{4t^2}{N^2}\bigg(\sum_{l\neq i}m_{l}^{(i)} m_{lj}^{(i)}  \bigg)^2 \, \frac{\big(4\sinh^2\big(h + \sum_{k\neq i} g_{ik}m_k^{(i)} \big)-1\big)}{\cosh^6\big(h + \sum_{k\neq i} g_{ik}m_k^{(i)} \big)} \bigg|\leq  \frac{C}{N^{1+\epsilon}}.\]	
This proves
		\[\begin{split}
		\mathbb{E}\, m_{ij}^2 =  &\, \mathbb{E} \,t  \,  \sech^{4} \Big(h + \sum_{k\neq i} g_{ik}m_k^{(i)} \Big)\frac1N \sum_{l\neq i} \big(m_{lj}^{(i)}\big)^2  +\Theta_4,
		\end{split}\]
for an error $ | \Theta_4| \leq C/N^{1+\epsilon}$, for any $\epsilon >0$ sufficiently small.

The rest of the proof of \eqref{eq:2normmij} follows now from a repetition of the arguments above. First, the concentration of $q_N^{(i)} $ (recall that $q_N^{(i)}$ and $q_N$ are close in $L^2(\mathbb{P})$) implies that
		\[\begin{split}
		\mathbb{E}\, m_{ij}^2 &=   \mathbb{E} \,t  \,  \sech^{4} \Big(h + \sum_{k\neq i} g_{ik}m_k^{(i)} \Big)\frac1N \sum_{l\neq i} \big(m_{lj}^{(i)}\big)^2  +\Theta_4\\
		&=   \mathbb{E} \,t  \,  \sech^{4} \big(h +  \sqrt{tq} Z \big)\, \mathbb{E}\,\frac1N \sum_{l\neq i} \big(m_{lj}^{(i)}\big)^2  +\Theta_5
		\end{split}\]
for $Z\sim \mathcal{N}(0,1)$ independent of the remaining disorder and an error $  | \Theta_5| \leq C/N^{1+\epsilon}$. Here, we have used the Lipschitz continuity of the map
		\[ x\mapsto \mathbb{E}_Z\sech^{4} \Big(h +  \sqrt{tx} Z \Big) \]
and, choosing $\delta>0$ sufficiently small, the bound
		\[\begin{split} &\mathbb{E}  \, \mathbb{E}_Z\,|m_{lj}^{(i)} |^2\Big[ \sech^{4} \Big(h +  \sqrt{tq_N^{(i)}} Z \Big) -  \sech^{4} \Big(h +  \sqrt{tq} Z \Big)  \Big] \\
		&\hspace{6cm} \leq C \big\|m_{lj}^{(i)}\big\|_{2+\delta}^{2} \big \|q_N^{(i)}-q\|_{(2+\delta)/\delta}^{\delta/(2+\delta)}\leq \frac{C}{N^{1+\epsilon}}
		  \end{split}\]
for $\epsilon=\epsilon_{\delta}>0$ small enough, by Lemma \ref{lm:L2bndmij} and Eq. \eqref{eq:concentrationLp} (applied to $q_N^{(i)}$).

Replacing then the $m_{lj}^{(i)} $ by $m_{lj}$ through It\^o's Lemma as in Section \ref{sec:hTAP} and using symmetry over the sites shows that
		\[\mathbb{E}\, m_{ij}^2 = \frac1N \mathbb{E} \,t  \,  \sech^{4} \big(h +  \sqrt{tq} Z \big) \mathbb{E}\, (1-m_j^2)^2 + \mathbb{E} \,t  \,  \sech^{4} \big(h +  \sqrt{tq} Z \big)\,\mathbb{E}\,m_{ij}^2 + \Theta_6 \]
for an error $ |\Theta_6| \leq C/N^{1+\epsilon}$. Finally, since
		\[ \Big| \mathbb{E}\, (1-m_j^2)^2 - \mathbb{E}   \,  \sech^{4} \Big(h +  \sqrt{tq} Z \Big)\Big|\leq  \frac{C}{N^{1/4}} \]
by the TAP equations \eqref{eq:hTAP1} and very similar arguments as above, we conclude
		\[\mathbb{E}\, m_{ij}^2 = \frac{t}N \bigg[1- \mathbb{E}  \,  \frac{t}{\cosh^{4} \big(h +  \sqrt{tq} Z \big) }\bigg]^{-1}  \bigg[ \mathbb{E}  \,  \frac{1}{\cosh^{4} \big(h +  \sqrt{tq} Z \big) }\bigg]^2   + \Theta_7\]
for an error $ |\Theta_7| \leq C/N^{1+\epsilon}$. This concludes the proof of Proposition \ref{prop:mij2norm}.
\end{proof}

\end{document}